\documentclass[11pt,oneside,reqno]{amsart}
\usepackage{graphicx} 
\usepackage{amsmath,amssymb,amsfonts,amsthm}
\usepackage[utf8]{inputenc}
\usepackage[T1]{fontenc}
\usepackage{natbib}
\usepackage{comment}
\usepackage{color}
\usepackage{mathrsfs}
\usepackage{lipsum}
\usepackage{subcaption}
\usepackage[skip=5pt]{caption}
\pdfoutput=1

\newtheorem{theorem}{Theorem}
\newtheorem{remark}{Remark}
\newtheorem{lemma}{Lemma}

\title{Phase-Space Analysis of Elastic Vector Solitons in Flexible Mechanical Metamaterials}
\usepackage{a4wide}
\author{M. H. Duong$^{1}$ and M. J. Reynolds$^{2}$ }
\address{$^{1}$ School of Mathematics, University of Birmingham, Birmingham, UK}
\address{$^{2}$ School of Engineering, University of Birmingham, Birmingham, UK}

\date{}

\begin{document}

\maketitle

\begin{abstract}
The purpose of this paper is to propose a revised continuum model from the discrete system introduced in \cite{deng2017elastic}. Using a Galilean transformation, we obtain an equation governing the soliton solutions in the phase plane - a second-order nonlinear ODE related to the Klein-Gordon equation with quadratic nonlinearity. These admit the well-known $\mathrm{sech}^2$ solutions, which we employ as an ans\"atz following \cite{deng2017elastic}. The resulting analysis yields soliton amplitudes and velocities that agree closely with numerical simulations, achieving an improvement of exactly 1/9 relative to the benchmark reported by the Harvard group.
\end{abstract}

\section{Introduction}


Metamaterials are artificially structured materials engineered to exhibit electromagnetic, acoustic, or mechanical properties that are not found in nature. Their remarkable behavior arises from subwavelength-scale structural design rather than chemical composition, allowing precise control of wave propagation and interaction. Mathematical modeling plays a fundamental role in understanding and predicting these effective macroscopic behaviors. Using homogenization theory and effective medium approximations, researchers derive constitutive parameters—such as effective permittivity, permeability, or stiffness—linking microscopic geometry to macroscopic response. Classical models rely on asymptotic and multiscale methods, while advanced formulations account for nonlocal effects, spatial dispersion, and local resonances. For instance, enriched continuum theories such as micropolar and micromorphic models have been developed to capture micro-rotational and gradient effects in mechanical metamaterials \cite{sridhar2016homogenization}. Nonlocal homogenization approaches extend these frameworks to include frequency- and wavevector-dependent constitutive relations, explaining phenomena such as artificial chirality and negative refractive indices \cite{ciattoni2015nonlocal}. More recent efforts integrate numerical homogenization, optimization algorithms, and data-driven design to model tunable, reconfigurable, and programmable metamaterials with dynamic control of wave and field behavior \cite{steckiewicz2022homogenization, ariza2024homogenization}. These mathematical and computational advances provide a rigorous foundation for designing next-generation metamaterials with tailored and multifunctional responses across various physical domains.

In~\cite{deng2017elastic,deng2017supplemental}, the authors analyse a purpose-built structure designed as a nonlinear flexible mechanical metamaterial. Their system comprised a $3 \times 20$ array of unit cells, each consisting of four rigid squares based on the geometric model of Grima and Evans~\cite{grima2000auxetic}. Each square possesses two degrees of freedom: a translational displacement in the $x$-direction and a rotation about the $z$-axis through its centre. The coupling between these degrees of freedom is governed by buckling at an internal angle of $\mathrm{25}^{\circ}$. The squares are interconnected by linear tension/compression and torsional springs, so that the nonlinearity of the system arises solely from the underlying geometry of the metamaterial - in fact the rotational part. The aforementioned papers investigated this system using three complementary approaches:
\begin{itemize}
    \item \textit{Experimental}: impact of the structure by a pendulum, with subsequent monitoring of the response;
    \item \textit{Numerical}: discrete modelling of the system integrated with a fourth-order Runge--Kutta scheme;
    \item \textit{Analytical}: derivation of a continuum approximation, retaining terms up to second order.
\end{itemize}


The present work is mainly concerned with the analytical approach. The purpose of this paper is to derive a revised continuum model from the discrete system introduced in \cite{deng2017elastic}, improving the continuum model obtained there by consistently including all terms up to second order. The resulting limiting system will be a coupled pair of nonlinear partial differential equations for the displacement and rotation. By employing a Galilean transformation, we derive an equation for the soliton solutions, which is a second-order nonlinear ODE. We then study the phase plane dynamics of the ODE, obtaining explicit solutions in some special cases.

The rest of the paper is organized as follows. In Section \ref{sec: discrete}
we revisit the discrete model introduced in \cite{deng2017elastic}. In Section \ref{sec: continuum limit} we derive a continuum model from the discrete one. In Section \ref{sec: soliton} we study soliton solutions for the continuum model. In Section \ref{sec: numerics}, we present some numerical investigations. A direct comparison with the results in \cite{deng2017elastic} is given in Section \ref{sec: comparison}. Finally, detailed computations are given in the appendix.
\section{Discrete model}
\label{sec: discrete}
We revisit the discrete model for soft-architected materials introduced in \cite{deng2017elastic}. We briefly recall the derivation of the governing discrete equations here for further analysis later and refer to \cite{deng2017elastic} for the detailed construction of the model.  

The structure consists of a network of square domains connected by thin ligaments, all fabricated from elastomeric material (polydimethylsiloxane — PDMS). The squares have diagonal lengths of $2l$ and are rotated by an angle $\theta_0$ relative to the undeformed diamond (rhombic) configuration. The propagation of plane waves along the $x$-direction is investigated.

To efficiently model the system, it is observed that when a planar wave propagates through the structure, all deformation localizes at the hinges, which bend in-plane and induce pronounced rotations of the squares. Consequently, the structure can be modeled as a network of rigid squares connected by springs at their vertices. Specifically, each hinge is represented by two linear springs: (i) a compression/tension spring with stiffness $k$, and (ii) a torsional spring with stiffness $k_\theta$.

Furthermore, when a planar wave propagates in the $x$-direction, the following behaviors are noted: (i) the squares do not move in the $y$-direction; (ii) vertically aligned neighboring squares undergo the same horizontal displacement and rotate by the same amount but in opposite directions; and (iii) neighboring squares always rotate in opposite directions. Since the focus is on the propagation of planar waves in the $x$-direction, each rigid square in the discrete model possesses two degrees of freedom: displacement in the $x$-direction, $u$, and rotation about the $z$-axis, $\theta$. Focusing on the rigid $[i,j]$-th square, the following relations hold:

$$
u_{i,j} = u_{i+1,j}, \quad \theta_{i,j} = \theta_{i+1,j}. 
$$
The normalized displacement $
U_{i,j} = \frac{u_{i,j}}{2l \cos \theta_0}$, time $T = t \sqrt{\frac{k_\ell p}{m}}$, stiffness $K = \frac{k_\theta}{k l^2}$, and inertia $\alpha = \frac{l p m}{J}$
Since only the displacements and rotations of squares in the $i$-th column play a role, for simplicity the notation is reduced to $U_j = U_{i,j}$ and $\theta_j = \theta_{i,j}$.
For completeness, we reproduce the full derivation of the discrete models from Newtonian dynamics in \cite{deng2017elastic} in the appendix. 

The resulting governing equations in a dimensionless form is given by
\begin{subequations}
    \begin{align}
\frac{\partial^2 U_j}{\partial T^2}&=U_{j+1}-2U_j+U_{j-1}-\frac{1}{2\cos(\theta_0)}\big[\cos(\theta_{j+1}+\theta_0)-\cos(\theta_{j-1}+\theta_0)+K(\theta_{j+1}-\theta_{j-1})\sin(\theta_j+\theta_0)\big]\label{eq: Uj},\\
\frac{\partial^2 \theta_j}{\partial T^2}&=\alpha^2\Bigg\{ -K(\theta_{j+1}+6\theta_j+\theta_{j-1})-2(U_{j+1}-U_{j-1})\cos(\theta_0)\sin(\theta_j+\theta_0)\notag\\
&\qquad +\sin(\theta_j+\theta_0)\Big[\cos(\theta_{j+1}+\theta_0)+6\cos(\theta_j+\theta_0)+\cos(\theta_{j-1}+\theta_0)-8\cos(\theta_0)\Big]\notag\\
&\qquad +\cos(\theta_j+\theta_0)\Big[\sin(\theta_{j+1}+\theta_0)+\sin(\theta_{j-1}+\theta_0)-2\sin(\theta_j+\theta_0)\Big]\Bigg\}.\label{eq: thetaj}
\end{align}
\end{subequations}
In the next section, we will derive the continuum limit from the above discrete system.
\section{Derivation of the continuum limit}
\label{sec: continuum limit}
We consider two continuous functions $U(X)$ and $\theta(X)$ that interpolate the discrete values $U_j$ and $\theta_j$ 
\[
U(X_j)=U_j,\quad \theta(X_j)=\theta_j.
\]
We use the following approximation of the first and second derivatives
\begin{subequations}
\label{eq: derivative approximation}
\begin{align}
\frac{\partial U}{\partial X}(X_j)&\approx \frac{1}{2}[U(X_{j+1})-U(X_{j-1})]=\frac{1}{2}(U_{j+1}-U_{j-1}),\label{eq: U1}
\\\frac{\partial^2 U}{\partial X^2}(X_j)&\approx U(X_{j+1})-2U(X_j)+U(X_{j-1})=U_{j+1}-2U_j+U_{j-1},\label{eq: U2}\\
\frac{\partial \theta}{\partial X}(X_j)&\approx \frac{1}{2}[\theta(X_{j+1})-\theta(X_{j-1}))]=\frac{1}{2}(\theta_{j+1}-\theta_{j-1}),\label{eq:theta1}
\\\frac{\partial^2 \theta}{\partial X^2}(X_j)&\approx \theta(X_{j+1})-2\theta(X_j)+\theta(X_{j-1})=\theta_{j+1}-2\theta_j+\theta_{j-1}.\label{eq:theta2}
\end{align}    
\end{subequations}

We also use the following approximations for $\sin(z)$ and $\cos(z)$ functions for small $z$
\[
\sin(z)\approx z\quad\text{and}\quad \cos(z)\approx 1-\frac{z^2}{2}.
\]
We assume that the rotation angle $\theta_j$ is small, to get
\begin{align}
\sin(\theta_j+\theta_0)&=\cos(\theta_j)\sin(\theta_0)+\sin(\theta_j)\cos(\theta_0)\approx \Big(1-\frac{\theta_j^2}{2}\Big)\sin(\theta_0)+\theta_j\cos(\theta_0),\label{eq: approx sin}\\
\cos(\theta_j+\theta_0)&=\cos(\theta_j)\cos(\theta_0)-\sin(\theta_j)\sin(\theta_0)\approx \Big(1-\frac{\theta_j^2}{2}\Big)\cos(\theta_0)-\theta_j\sin(\theta_0).\label{eq: approx cos}
\end{align}
We now apply the approximations to the right hand side expressions in the discrete equations \eqref{eq: Uj}-\eqref{eq: thetaj} in order to derive the corresponding continuum limits. 
\subsection{Continuum equation for the displacement}
We start with \eqref{eq: Uj} to derive the continuum equation for the displacement. We consider the first two terms inside the square bracket:
\begin{align}
\cos(\theta_{j+1}+\theta_0)-\cos(\theta_{j-1}+\theta_0)&\approx \Big(1-\frac{\theta_{j+1}^2}{2}\Big)\cos(\theta_0)-\theta_{j+1}\sin(\theta_0)-\Big(1-\frac{\theta_{j-1}^2}{2}\Big)\cos(\theta_0)+\theta_{j-1}\sin(\theta_0)\notag
 \\&=-\frac{1}{2}(\theta_{j+1}^2-\theta_{j-1}^2)\cos(\theta_0)-(\theta_{j+1}-\theta_{j-1})\sin(\theta_0).\label{eq: diff1}
\end{align}
To proceed, we approximate the term $\theta_{j+1}^2-\theta_{j-1}^2$. By multiplying \eqref{eq:theta1} then adding with \eqref{eq:theta2} we obtain the following approximation
\begin{equation}
\label{eq: approx-sq1}
\theta_{j+1}\approx \theta_j+\frac{\partial\theta}{\partial X}(X_j)+\frac{1}{2}\frac{\partial^2\theta}{\partial X^2}(X_j).
\end{equation}
Similarly, by multiplying \eqref{eq:theta1} then subtracting \eqref{eq:theta2}, we get
\begin{equation}
\label{eq: approx-sq2}
\theta_{j-1}\approx \theta_{j}-\frac{\partial\theta}{\partial X}(X_j)+\frac{1}{2}\frac{\partial^2\theta}{\partial X^2}(X_j).
\end{equation}
Thus, using the simple identity $(a+b)^2-(a-b)^2=4 ab$, we have
\begin{align}
  \theta_{j+1}^2-\theta_{j-1}^2&\approx \Big(\theta_j+\frac{\partial\theta}{\partial X}(X_j)+\frac{1}{2}\frac{\partial^2\theta}{\partial X^2}(X_j)\Big)^2-\Big(\theta_{j}-\frac{\partial\theta}{\partial X}(X_j)+\frac{1}{2}\frac{\partial^2\theta}{\partial X^2}(X_j)\Big)^2 \notag
  \\&= 4\Big(\theta(X_j)+\frac{1}{2}\frac{\partial^2\theta}{\partial X^2}(X_j)\Big)\frac{\partial\theta}{\partial X}(X_j)\notag
  \\&=4\theta(X_j)\frac{\partial\theta}{\partial X}(X_j)+2\frac{\partial^2\theta}{\partial X^2}(X_j)\frac{\partial\theta}{\partial X}(X_j)\notag
  \\&=2\frac{\partial\theta^2}{\partial X}(X_j)+\frac{\partial}{\partial X}\Big[\Big(\frac{\partial\theta}{\partial X}\Big)^2\Big](X_j).\label{eq: diff2}
\end{align}
Substituting \eqref{eq:theta1} and \eqref{eq: diff2} to \eqref{eq: diff1} we have
\begin{align}
\cos(\theta_{j+1}+\theta_0)-\cos(\theta_{j-1}+\theta_0)\approx -\frac{1}{2}\Bigg(2\frac{\partial\theta^2}{\partial X}(X_j)+\frac{\partial}{\partial X}\Big[\Big(\frac{\partial\theta}{\partial X}\Big)^2\Big](X_j)\Bigg)\cos(\theta_0)-2\sin(\theta_0)\frac{\partial\theta}{\partial X}(X_j).\label{eq: diff3}
\end{align}
Next, for the second term inside the square bracket of the right-hand side of \eqref{eq: Uj}, and by keeping only the dominant terms, we have
\begin{align}
(\theta_{j+1}-\theta_{j-1})\sin(\theta_{j}+\theta_0)&\approx (\theta_{j+1}-\theta_{j-1})\Big((1-\frac{\theta_j^2}{2})\sin(\theta_0)+\theta_j\cos(\theta_0)\Big)\notag
\\& \approx 2\frac{\partial \theta}{\partial X}(X_j)\Big((1-\frac{1}{2}\theta(X_j)^2)\sin(\theta_0)+\theta(X_j)\cos(\theta_0)\Big)\notag
\\&= 2\frac{\partial \theta}{\partial X}(X_j)\sin(\theta_0)+\frac{\partial\theta^2}{\partial X}(X_j)\cos(\theta_0)-\frac{1}{3}\frac{\partial\theta^3}{\partial X}(X_j)\sin(\theta_0)\notag
\\&\approx 2\frac{\partial \theta}{\partial X}(X_j)\sin(\theta_0)+\frac{\partial\theta^2}{\partial X}(X_j)\cos(\theta_0).\label{eq: diff4}
\end{align}
Substituting \eqref{eq: U2}, \eqref{eq: diff3} and \eqref{eq: diff4} into \eqref{eq: Uj}, we get
\begin{align*}
\frac{\partial^2 U}{\partial T^2}(X_j)&\approx  \frac{\partial^2 U}{\partial X^2}(X_j)-\frac{1}{2\cos(\theta_0)}\Bigg\{-\Big(\frac{\partial\theta^2}{\partial X}(X_j)+\frac{1}{2}\frac{\partial}{\partial X}\Big[\Big(\frac{\partial\theta}{\partial X}\Big)^2\Big](X_j)\Big)\cos(\theta_0)-2\sin(\theta_0)\frac{\partial\theta}{\partial X}(X_j)
\\&\quad + 2K\sin(\theta_0)\frac{\partial \theta}{\partial X}(X_j)+K\frac{\partial\theta^2}{\partial X}(X_j)\cos(\theta_0)\Bigg\}
\\&=\frac{\partial^2 U}{\partial X^2}(X_j)+(1-K)\tan(\theta_0)\frac{\partial \theta}{\partial X}(X_j)+\frac{1}{2}(1-K)\frac{\partial\theta^2}{\partial X}(X_j)+\frac{1}{4}\frac{\partial}{\partial X}\Big[\Big(\frac{\partial\theta}{\partial X}\Big)^2\Big](X_j).
\end{align*}
This leads to the following continuum equation for the displacement $U=U(T,X)$:
\begin{equation}
 \frac{\partial^2 U}{\partial T^2}(X)= \frac{\partial^2 U}{\partial X^2}(X)+(1-K)\tan(\theta_0)\frac{\partial \theta}{\partial X}(X)+\frac{1}{2}(1-K)\frac{\partial\theta^2}{\partial X}(X)+\frac{1}{4}\frac{\partial}{\partial X}\Big[\Big(\frac{\partial\theta}{\partial X}\Big)^2\Big](X).  \label{eq: continuum U}
\end{equation}
Remark: if we keep the remaining term in \eqref{eq: diff4} then \eqref{eq: continuum U} becomes
\[
 \frac{\partial^2 U}{\partial T^2}(X)= \frac{\partial^2 U}{\partial X^2}(X)+(1-K)\tan(\theta_0)\frac{\partial \theta}{\partial X}(X)+\frac{1}{2}(1-K)\frac{\partial\theta^2}{\partial X}(X)+\frac{1}{4}\frac{\partial}{\partial X}\Big[\Big(\frac{\partial\theta}{\partial X}\Big)^2\Big](X)+\frac{1}{6}K\tan(\theta_0)\frac{\partial \theta^3}{\partial X}.
\]
\subsection{Continuum equation for the rotation}
We denote
\begin{align}
 A&:= -K(\theta_{j+1}+6\theta_j+\theta_{j-1}),\\
 B&:=-2(U_{j+1}-U_{j-1})\cos(\theta_0)\sin(\theta_j+\theta_0),\\
 C&:=\sin(\theta_j+\theta_0)\Big[\cos(\theta_{j+1}+\theta_0)+6\cos(\theta_j+\theta_0)+\cos(\theta_{j-1}+\theta_0)-8\cos(\theta_0)\Big],\\
 D&:=\cos(\theta_j+\theta_0)\Big[\sin(\theta_{j+1}+\theta_0)+\sin(\theta_{j-1}+\theta_0)-2\sin(\theta_j+\theta_0)\Big].
\end{align}
so that the discrete equation for the rotations \eqref{eq: thetaj} can be written as
\begin{equation}
\label{eq: thetaj2}
\frac{\partial^2\theta_j}{\partial T^2}=\alpha^2(A+B+C+D).    
\end{equation}
i) For the term $A$, using the approximation \eqref{eq:theta2}, we have
\begin{align}
A&=-K(\theta_{j+1}+6\theta_j+\theta_{j-1}) \notag\\&= -K(\theta_{j+1}-2\theta_j+\theta_{j-1}+8\theta_j)\notag
\\&\approx -K\Big[\frac{\partial^2\theta}{\partial X^2}(X_j)+8\theta(X_j)\Big].\label{eq: A}
\end{align}
ii) For the term $B$, applying the approximations \eqref{eq: U2} and \eqref{eq: approx sin}, we get
\begin{align}
 B&\approx -4\frac{\partial U}{\partial X}(X_j)\cos(\theta_0)\Big[ \Big(1-\frac{\theta_j^2}{2}\Big)\sin(\theta_0)+\theta_j\cos(\theta_0)\Big]\notag
 \\&=-4\frac{\partial U}{\partial X}(X_j)\cos(\theta_0)\sin(\theta_0)-4\frac{\partial U}{\partial X}(X_j)\theta_j\cos(\theta_0)^2+2\frac{\partial U}{\partial X}(X_j)\theta_j^2\cos(\theta_0)\sin(\theta_0)\notag
\\& \approx -2\frac{\partial U}{\partial X}(X_j)\sin(2\theta_0)-4\cos(\theta_0)^2\,\theta_j\,\frac{\partial U}{\partial X}(X_j).\label{eq: B}
\end{align}
iii) We proceed to $C$. For the term inside the square bracket, from \eqref{eq: approx cos} we have
\begin{align}
\label{eq: C0}
&\cos(\theta_{j+1}+\theta_0)+6\cos(\theta_j+\theta_0)+\cos(\theta_{j-1}+\theta_0)-8\cos(\theta_0)\notag
\\&\qquad\approx  \Big(1-\frac{\theta_{j+1}^2}{2}\Big)\cos(\theta_0)-\theta_{j+1}\sin(\theta_0)+6\Big[\Big(1-\frac{\theta_j^2}{2}\Big)\cos(\theta_0)-\theta_j\sin(\theta_0)\Big]\notag
\\&\qquad\qquad +\Big(1-\frac{\theta_{j-1}^2}{2}\Big)\cos(\theta_0)-\theta_{j-1}\sin(\theta_0)-8\cos(\theta_0)\notag
\\&=-\frac{1}{2}\cos(\theta_0)(\theta_{j+1}^2+6\theta_{j}^2+\theta_{j-1}^2)-\sin(\theta_0)(\theta_{j+1}+6\theta_j+\theta_{j-1}).
\end{align}
We now approximate the term $\theta_{j+1}^2+6\theta_j^2+\theta_{j-1}^2$ appearing in the last expression. To this end, using the relations \eqref{eq: approx-sq1}-\eqref{eq: approx-sq2}, and the elementary identity $(a+b)^2+(a-b)^2=2(a^2+b^2)$, we have
\begin{align}
\theta_{j+1}^2+6\theta_j^2+\theta_{j-1}^2&\approx\Big(\theta_j+\frac{\partial\theta}{\partial X}(X_j)+\frac{1}{2}\frac{\partial^2\theta}{\partial X^2}(X_j)\Big)^2+  6\theta_j^2+  \Big(\theta_j-\frac{\partial\theta}{\partial X}(X_j)+\frac{1}{2}\frac{\partial^2\theta}{\partial X^2}(X_j)\Big)^2\notag
\\&=2\Big[\Big(\theta_j+\frac{1}{2}\frac{\partial^2\theta}{\partial X^2}(X_j)\Big)^2+\Big(\frac{\partial\theta}{\partial X}(X_j)\Big)^2\Big]+6\theta_j^2\notag
\\&=8\theta_j^2+2\theta_j \frac{\partial^2\theta}{\partial X^2}(X_j)+\frac{1}{2}\Big(\frac{\partial^2\theta}{\partial X^2}(X_j)\Big)^2+2\Big(\frac{\partial\theta}{\partial X}(X_j)\Big)^2\notag\\
&=8\theta(X_j)^2+2\theta(X_j) \frac{\partial^2\theta}{\partial X^2}(X_j)+\frac{1}{2}\Big(\frac{\partial^2\theta}{\partial X^2}(X_j)\Big)^2+2\Big(\frac{\partial\theta}{\partial X}(X_j)\Big)^2.\label{eq: C1}
\end{align}
Next, as in term $A$, we have
\begin{align}
\theta_{j+1}+6\theta_j+\theta_{j-1}\approx \frac{\partial^2\theta}{\partial X^2}(X_j)+8\theta(X_j).  \label{eq: C2}
\end{align}
Substituting \eqref{eq: C1} and \eqref{eq: C2} into \eqref{eq: C0} we get
\begin{align}
&\cos(\theta_{j+1}+\theta_0)+6\cos(\theta_j+\theta_0)+\cos(\theta_{j-1}+\theta_0)-8\cos(\theta_0)\notag
\\&\quad\approx -\frac{1}{2}\cos(\theta_0)\Big[8\theta(X_j)^2+2\theta(X_j) \frac{\partial^2\theta}{\partial X^2}(X_j)+\frac{1}{2}\Big(\frac{\partial^2\theta}{\partial X^2}(X_j)\Big)^2+2\Big(\frac{\partial\theta}{\partial X}(X_j)\Big)^2\Big]-\sin(\theta_0)\Big[\frac{\partial^2\theta}{\partial X^2}(X_j)+8\theta(X_j)\Big]\notag
\\&\quad=-\cos(\theta_0)\Big[4\theta(X_j)^2+\theta(X_j) \frac{\partial^2\theta}{\partial X^2}(X_j)+\frac{1}{4}\Big(\frac{\partial^2\theta}{\partial X^2}(X_j)\Big)^2+\Big(\frac{\partial\theta}{\partial X}(X_j)\Big)^2\Big]-\sin(\theta_0)\Big[\frac{\partial^2\theta}{\partial X^2}(X_j)+8\theta(X_j)\Big].\label{eq: C3}
\end{align}
Substituting \eqref{eq: approx sin} and \eqref{eq: C3} into $C$ we get
\begin{align}
C&\approx\Big[\Big(1-\frac{\theta_j^2}{2}\Big)\sin(\theta_0)+\theta_j\cos(\theta_0)\Big]\bigg\{-\cos(\theta_0)\Big[4\theta(X_j)^2+\theta(X_j) \frac{\partial^2\theta}{\partial X^2}(X_j)+\frac{1}{4}\Big(\frac{\partial^2\theta}{\partial X^2}(X_j)\Big)^2+\Big(\frac{\partial\theta}{\partial X}(X_j)\Big)^2\Big]\notag
\\& \hspace{6cm}-\sin(\theta_0)\Big[\frac{\partial^2\theta}{\partial X^2}(X_j)+8\theta(X_j)\Big]\bigg\} \notag
\\&\approx -\frac{1}{2}\sin(2\theta_0)\Big[4\theta(X_j)^2+\theta(X_j) \frac{\partial^2\theta}{\partial X^2}(X_j)+\frac{1}{4}\Big(\frac{\partial^2\theta}{\partial X^2}(X_j)\Big)^2+\Big(\frac{\partial\theta}{\partial X}(X_j)\Big)^2\Big]\notag
\\&\qquad -\sin(\theta_0)^2\Big[\frac{\partial^2\theta}{\partial X^2}(X_j)+8\theta(X_j)\Big]-\frac{1}{2}\sin(2\theta_0)\theta(X_j)\Big[\frac{\partial^2\theta}{\partial X^2}(X_j)+8\theta(X_j)\Big]\\
\label{eq: C}
&=-\frac{1}{2}\sin(2\theta_0)\Big[12\theta(X_j)^2+2\theta(X_j) \frac{\partial^2\theta}{\partial X^2}(X_j)+\frac{1}{4}\Big(\frac{\partial^2\theta}{\partial X^2}(X_j)\Big)^2+\Big(\frac{\partial\theta}{\partial X}(X_j)\Big)^2\Big]\notag
\\&\qquad -\sin(\theta_0)^2\Big[\frac{\partial^2\theta}{\partial X^2}(X_j)+8\theta(X_j)\Big].
\end{align}
iv) Next, term $D$. Using \eqref{eq: approx sin} we have
\begin{align}
&\sin(\theta_{j+1}+\theta_0)+\sin(\theta_{j-1}+\theta_0)-2\sin(\theta_j+\theta_0)\notag
\\&\quad \approx \Big(1-\frac{\theta_{j+1}^2}{2}\Big)\sin(\theta_0)+\theta_{j+1}\cos(\theta_0)+\Big(1-\frac{\theta_{j-1}^2}{2}\Big)\sin(\theta_0)+\theta_{j-1}\cos(\theta_0)-2\Big[\Big(1-\frac{\theta_j^2}{2}\Big)\sin(\theta_0)+\theta_j\cos(\theta_0)\Big]\notag
\\&\quad =-\frac{1}{2}\sin(\theta_0)(\theta_{j+1}^2-2\theta_j^2+\theta_{j-1}^2)+\cos(\theta_0)(\theta_{j+1}-2\theta_j+\theta_{j-1}).\label{eq: D1}
\end{align}
Similarly as in \eqref{eq: C1} we have
\begin{align}
\theta_{j+1}^2-2\theta_j^2+\theta_{j-1}^2&\approx\Big(\theta_j+\frac{\partial\theta}{\partial X}(X_j)+\frac{1}{2}\frac{\partial^2\theta}{\partial X^2}(X_j)\Big)^2+  \Big(\theta_j-\frac{\partial\theta}{\partial X}(X_j)+\frac{1}{2}\frac{\partial^2\theta}{\partial X^2}(X_j)\Big)^2-2\theta_j^2\notag
\\&=2\Big[\Big(\theta_j+\frac{1}{2}\frac{\partial^2\theta}{\partial X^2}(X_j)\Big)^2+\Big(\frac{\partial\theta}{\partial X}(X_j)\Big)^2\Big]-2\theta_j^2\notag 
\\&=2\theta_j\frac{\partial^2\theta}{\partial X^2}(X_j)+\frac{1}{2}\Big(\frac{\partial^2\theta}{\partial X^2}(X_j)\Big)^2+2\Big(\frac{\partial\theta}{\partial X}(X_j)\Big)^2\notag
\\&=2\theta(X_j)\frac{\partial^2\theta}{\partial X^2}(X_j)+\frac{1}{2}\Big(\frac{\partial^2\theta}{\partial X^2}(X_j)\Big)^2+2\Big(\frac{\partial\theta}{\partial X}(X_j)\Big)^2.\label{eq: D2}
\end{align}
Substituting \eqref{eq:theta2} and \eqref{eq: D2} into \eqref{eq: D1} yields
\begin{align}
&\sin(\theta_{j+1}+\theta_0)+\sin(\theta_{j-1}+\theta_0)-2\sin(\theta_j+\theta_0)\\
&\approx-\frac{1}{2}\sin(\theta_0)\Big[2\theta(X_j)\frac{\partial^2\theta}{\partial X^2}(X_j)+\frac{1}{2}\Big(\frac{\partial^2\theta}{\partial X^2}(X_j)\Big)^2+2\Big(\frac{\partial\theta}{\partial X}(X_j)\Big)^2\Big]+\cos(\theta_0)\frac{\partial^2\theta}{\partial X^2}(X_j)\notag
\\&=-\sin(\theta_0)\Big[\theta(X_j)\frac{\partial^2\theta}{\partial X^2}(X_j)+\frac{1}{4}\Big(\frac{\partial^2\theta}{\partial X^2}(X_j)\Big)^2+\Big(\frac{\partial\theta}{\partial X}(X_j)\Big)^2\Big]+\cos(\theta_0)\frac{\partial^2\theta}{\partial X^2}(X_j).\label{eq: D3}
\end{align}
Substituting \eqref{eq: approx cos} and \eqref{eq: D3} into \eqref{eq: D1} we get
\begin{align}
D&\approx   \Big(\cos(\theta_0)-\theta_j\sin(\theta_0)-\frac{\theta_j^2}{2}\cos(\theta_0)\Big) \Big\{-\sin(\theta_0)\Big[\theta(X_j)\frac{\partial^2\theta}{\partial X^2}(X_j)+\frac{1}{4}\Big(\frac{\partial^2\theta}{\partial X^2}(X_j)\Big)^2+\Big(\frac{\partial\theta}{\partial X}(X_j)\Big)^2\Big]\notag
\\&\hspace{6cm}+\cos(\theta_0)\frac{\partial^2\theta}{\partial X^2}(X_j)\Big\}\notag
\\&\approx -\frac{1}{2}\sin(2\theta_0)\Big[\theta(X_j)\frac{\partial^2\theta}{\partial X^2}(X_j)+\frac{1}{4}\Big(\frac{\partial^2\theta}{\partial X^2}(X_j)\Big)^2+\Big(\frac{\partial\theta}{\partial X}(X_j)\Big)^2\Big]+\cos(\theta_0)^2\frac{\partial^2\theta}{\partial X^2}(X_j)
\notag\\&\qquad-\frac{1}{2}\theta(X_j)\sin(2\theta_0)\frac{\partial^2\theta}{\partial X^2}(X_j).\label{eq: D}
\end{align}
v) Summing all terms $A, B, C$ and $D$ from \eqref{eq: A}, \eqref{eq: B}, \eqref{eq: C} and \eqref{eq: D} we have
\begin{align}
A+B+C+D&\approx  -K\Big[\frac{\partial^2\theta}{\partial X^2}(X_j)+8\theta(X_j)\Big]-2\frac{\partial U}{\partial X}(X_j)\sin(2\theta_0)-4\cos(\theta_0)^2\,\theta_j\,\frac{\partial U}{\partial X}(X_j)\notag\\
&\qquad -\frac{1}{2}\sin(2\theta_0)\Big[12\theta(X_j)^2+2\theta(X_j) \frac{\partial^2\theta}{\partial X^2}(X_j)+\frac{1}{4}\Big(\frac{\partial^2\theta}{\partial X^2}(X_j)\Big)^2+\Big(\frac{\partial\theta}{\partial X}(X_j)\Big)^2\Big]\notag
\\&\qquad -\sin(\theta_0)^2\Big[\frac{\partial^2\theta}{\partial X^2}(X_j)+8\theta(X_j)\Big]\notag
\\&\qquad-\frac{1}{2}\sin(2\theta_0)\Big[\theta(X_j)\frac{\partial^2\theta}{\partial X^2}(X_j)+\frac{1}{4}\Big(\frac{\partial^2\theta}{\partial X^2}(X_j)\Big)^2+\Big(\frac{\partial\theta}{\partial X}(X_j)\Big)^2\Big]+\cos(\theta_0)^2\frac{\partial^2\theta}{\partial X^2}(X_j)\notag
\\&\qquad-\frac{1}{2}\sin(2\theta_0)\theta(X_j)\frac{\partial^2\theta}{\partial X^2}(X_j)\notag
\\&=(\cos(\theta_0)^2-\sin(\theta_0)^2-K)\frac{\partial^2\theta}{\partial X^2}(X_j)-2\sin(2\theta_0)\frac{\partial U}{\partial X}(X_j)\notag
\\&\qquad-4\theta(X_j)\Big(2K+\cos(\theta_0)^2\frac{\partial U}{\partial X}(X_j)+2\sin(\theta_0)^2\Big)\notag
\\&\qquad-\sin(2\theta_0)\Big[\frac{1}{4}\Big(\frac{\partial^2\theta}{\partial X^2}(X_j)\Big)^2+\Big(\frac{\partial\theta}{\partial X}(X_j)\Big)^2\Big]\notag 
\\&\qquad-\sin(2\theta_0)\Big[2\theta(X_j) \frac{\partial^2\theta}{\partial X^2}(X_j)+6\theta(X_j)^2\Big].
\end{align}
Substituting this into \eqref{eq: thetaj2}, we obtain the continuum approximation for the rotation
\begin{align}
\label{eq: continuum theta}
\frac{\partial^2\theta}{\partial T^2}(X)&=\alpha^2\Bigg\{(\cos(\theta_0)^2-\sin(\theta_0)^2-K)\frac{\partial^2\theta}{\partial X^2}(X)-2\sin(2\theta_0)\frac{\partial U}{\partial X}(X)-4\theta(X)\Big(2K+\cos(\theta_0)^2\frac{\partial U}{\partial X}(X_j)\notag
\\&\quad\qquad+2\sin(\theta_0)^2\Big)-\sin(2\theta_0)\Big[\frac{1}{4}\Big(\frac{\partial^2\theta}{\partial X^2}(X)\Big)^2+\Big(\frac{\partial\theta}{\partial X}(X)\Big)^2\Big]\notag
\\&\qquad\qquad-\sin(2\theta_0)\Big[2\theta(X) \frac{\partial^2\theta}{\partial X^2}(X)+6\theta(X)^2\Big]\Bigg\}.  \end{align}
In conclusion, we obtain the following continuum equations: coupled PDEs, for displacement and rotation:
\begin{subequations}
    \begin{align}
 \frac{\partial^2 U}{\partial T^2}(X)&= \frac{\partial^2 U}{\partial X^2}(X)+(1-K)\tan(\theta_0)\frac{\partial \theta}{\partial X}(X)+\mathbf{\frac{1}{2}(1-K)\frac{\partial\theta^2}{\partial X}(X)+\frac{1}{4}\frac{\partial}{\partial X}\Big[\Big(\frac{\partial\theta}{\partial X}\Big)^2\Big](X)},\label{eq: cU}\\
\frac{\partial^2\theta}{\partial T^2}(X)&=\alpha^2\Bigg\{(\cos(2\theta_0)-K)\frac{\partial^2\theta}{\partial X^2}(X)-2\sin(2\theta_0)\frac{\partial U}{\partial X}(X)-4\theta(X)\Big(2K+\cos(\theta_0)^2\frac{\partial U}{\partial X}(X)+2\sin(\theta_0)^2\Big)\notag
\\&\quad\qquad-4\sin(2\theta_0)\theta(X)^2\mathbf{-\sin(2\theta_0)\Big[\frac{1}{4}\Big(\frac{\partial^2\theta}{\partial X^2}(X)\Big)^2+\Big(\frac{\partial\theta}{\partial X}(X)\Big)^2+2\theta(X) \frac{\partial^2\theta}{\partial X^2}(X)+2\theta(X)^2\Big]\Bigg\}}. \label{eq: ctheta}
\end{align}
\end{subequations}
Note that, compared to \cite{deng2017elastic}\cite{deng2017supplemental}, the bold terms are newly introduced in the present paper. In the next section, we analyse the coupled-system, showing the effects of the new terms.
\section{Soliton solutions}
\label{sec: soliton}
In this section, we seek soliton solutions to the coupled continuum equations obtained in the previous section. We then study qualitatively the properties of the solutions. 

We use a Galilean transformation
\begin{equation}
    \zeta-\zeta_0=X-cT,
\end{equation}
and seek travelling-wave solutions for the coupled system \eqref{eq: cU}-\eqref{eq: ctheta} of the form
\[
U(T,X)=u(\zeta),\quad \theta(T,X)=\eta(\zeta).
\]
Using the chain rule, we have
\begin{align*}
\frac{\partial U}{\partial T}(T,X)&=\frac{du}{d\zeta}(\zeta)\frac{\partial\zeta}{\partial T}=-c\frac{du}{d\zeta}(\zeta),\quad 
\frac{\partial^2 U}{\partial T^2}(T,X)=c^2\frac{d^2u}{d\zeta}(\zeta),\\
\frac{\partial U}{\partial X}(T,X)&=\frac{du}{d\zeta}(\zeta)\frac{\partial\zeta}{\partial X}=\frac{du}{d\zeta}(\zeta),\quad 
\frac{\partial^2 U}{\partial X^2}(T,X)=\frac{d^2u}{d\zeta^2}(\zeta),\\
\frac{\partial \theta}{\partial T}(T,X)&=\frac{d\eta}{d\zeta}(\zeta)\frac{\partial\zeta}{\partial T}=-c\frac{d\eta}{d\zeta}(\zeta),\quad 
\frac{\partial^2 \theta}{\partial T^2}(T,X)=c^2\frac{d^2\eta}{d\zeta}(\zeta),\\
\frac{\partial \theta}{\partial X}(T,X)&=\frac{d\eta}{d\zeta}(\zeta)\frac{\partial\zeta}{\partial X}=\frac{d\eta}{d\zeta}(\zeta),\quad 
\frac{\partial^2 \theta}{\partial X^2}(T,X)=\frac{d^2\eta}{d\zeta^2}(\zeta),\\
\frac{\partial \theta^2}{\partial X}&=2\theta\frac{\partial\theta}{\partial X}=2\eta(\zeta)\frac{d\eta}{d\zeta}(\zeta),\\
\frac{\partial}{\partial X}\Big[\Big(\frac{\partial\theta}{\partial X}\Big)^2\Big](X)&=2\frac{\partial\theta}{\partial X}\frac{\partial^2\theta}{\partial X^2}=2\frac{d\eta}{d\zeta}(\zeta)\frac{d^2\eta}{d\zeta^2}(\zeta).
\end{align*}
Substituting these expressions into \eqref{eq: cU}-\eqref{eq: ctheta}, we obtain
\begin{subequations}
\begin{align}
 \frac{d^2 u}{d\zeta^2}(\zeta)&=-\frac{(1-K)}{1-c^2}\tan(\theta_0)\frac{d\eta}{d\zeta}(\zeta)-\frac{(1-K)}{1-c^2}\eta(\zeta)\frac{d\eta}{d\zeta}(\zeta)-\frac{1}{2(1-c^2)} \frac{d\eta}{d\zeta}(\zeta)\frac{d^2\eta}{d\zeta^2}(\zeta),\label{eq: cu2}\\
c^2\frac{d^2\eta}{d\zeta^2}(\zeta)&=\alpha^2\Big\{(\cos(2\theta_0)-K)\frac{d^2\eta}{d\zeta^2}(\zeta)-2\sin(2\theta_0)\frac{du}{d\zeta}(\zeta)-4\eta(\zeta)\big(2K+\cos(\theta_0)^2\frac{du}{d\zeta}(\zeta)+2\sin(\theta_0)^2\big)\notag
\\&\quad-4\sin(2\theta_0)\eta(\zeta)^2-\sin(2\theta_0)\Big[\frac{1}{4}\Big(\frac{d^2\eta}{d\zeta^2}(\zeta)\Big)^2+\Big(\frac{d\eta}{d\zeta}(\zeta)\Big)^2+2\eta(\zeta)\frac{d^2\eta}{d\zeta^2}(\zeta)+2\eta(\zeta)^2\Big]
\Big\}\label{eq: ceta2}.
 \end{align}  
\end{subequations}
Integrating with respect to $\zeta$ in \eqref{eq: cu2} we get
\begin{equation*}
 \frac{d u}{d\zeta}(\zeta)=-\frac{(1-K)}{1-c^2}\tan(\theta_0)\eta-\frac{(1-K)}{2(1-c^2)}\eta^2-\frac{1}{4(1-c^2)} \Big(\frac{d\eta}{d\zeta}(\zeta)\Big)^2.
\end{equation*}
Thus the evolution of the displacement is fully governed by the rotation. It is therefore crucial to analyse $\eta$. To this end, we derive an equation for $\eta$.

Substituting the above expression to \eqref{eq: ceta2} and ignoring higher-order terms, we obtain
\begin{align}
&\Big(c^2-\alpha^2(\cos(2\theta_0)-K-2\sin(2\theta_0)\eta)\Big)\frac{d^2\eta}{d\zeta^2}(\zeta)\notag
\\&\quad\approx\alpha^2\Big\{\Big[\frac{2(1-K)}{1-c^2}\tan(\theta_0)\sin(2\theta_0)-4(2K+2\sin(\theta_0)^2)\Big]\eta\notag
\\&\qquad +\Big[\frac{1-K}{1-c^2}(\sin(2\theta_0)+4\cos(\theta_0)^2\tan(\theta_0))-6\sin(2\theta_0)\Big]\eta^2\\&\qquad+\Big[\frac{1}{2(1-c^2)}\sin(2\theta_0)-\sin(2\theta_0)\Big]\Big(\frac{d\eta}{d\zeta}(\zeta)\Big)^2-\frac{1}{4}\sin(2\theta_0)\Big(\frac{d^2\eta}{d\zeta^2}(\zeta)\Big)^2\Big\}.\label{eq: ceta3}
\end{align}
We use the following identities
\begin{align*}
\tan(\theta_0)\sin(2\theta_0)&=2\sin^2(\theta_0).
\end{align*}
\begin{align*}
4\cos^2(\theta_0)\tan(\theta_0)&=2\sin(2\theta_0)
\end{align*}
We define
\begin{align*}
 \beta(\eta)&=A\eta+B= \alpha^2(\cos(2\theta_0)-K-2\sin(2\theta_0)\eta-c^2/\alpha^2),
 \\ P&=8\frac{(c^2-K/2-1/2)}{(1-c^2)}\sin^2(\theta_0)-8K
 \\ Q&=6\frac{(c^2-K/2-1/2)}{(1-c^2)}\sin(2\theta_0)\\
 R&=\frac{(c^2-1/2)}{(1-c^2)}\sin(2\theta_0)\\
 S&=-\frac{1}{4}\sin(2\theta_0).
\end{align*}
Then \eqref{eq: ceta3} can be written as
\begin{equation}
\label{eq: ceta4}
\beta(\eta)\eta''+
\alpha^2(P\eta +Q \eta^2+R(\eta')^2+S(\eta'')^2)=0,\quad\text{where}~ \eta'=\frac{d\eta}{d\zeta}(\zeta),~\eta''=\frac{d^2\eta}{d\zeta^2}(\zeta).   \end{equation}
We define
\begin{equation}
 \label{eq: coefficients}   
A=-2\alpha^2\sin(2\theta_0), \quad B=\alpha^2(\cos(2\theta_0)-K)-c^2,
\quad C=\alpha^2 P,\quad D=\alpha^2 Q, \quad E=\alpha^2 R, \quad F=\alpha^2 S. 
\end{equation}
Then \eqref{eq: ceta4} becomes
\begin{equation}
\label{eq: ceta5}
(A \eta+B)\eta''+ C \eta  +D\eta ^2 +E (\eta')^2+F (\eta'')^2=0.
\end{equation}
\subsection{Phase space dynamics}
We consider \eqref{eq: ceta5} in the case of $F=0$, ie the following second-order nonlinear ordinary differential equation:
\begin{equation}
\label{eq: ceta6}
(A \eta + B)\eta'' + C \eta + D\eta^2 + E (\eta')^2 = 0,
\end{equation}
where $\eta = \eta(\zeta)$, and $A, B, C, D, E$ are real constants.\\ 

Our aim in this section is to reduce \eqref{eq: ceta6} to a first-order ODE using an integrating factor method. A similar approach has been investigated in \cite{Jordan2011nonlinear} \cite{Jordan2007sourcebook}. To this end, we will need the following technical lemma.
\begin{lemma}
\label{lem: integration}
For $\sigma\ne -3,-2,-1$ we have
\begin{align*}
&\int (Ax + B)^\sigma (C x + D x^2) \, dx
\\ 
&\qquad=\frac{1}{A^2} \Bigg[
\frac{D}{A(\sigma+3)} (Ax + B)^{\sigma+3}
+ \frac{AC-2DB}{A(\sigma+2)}(Ax + B)^{\sigma+2}
+ \frac{B(DB-AC)}{A(\sigma+1)}(Ax + B)^{\sigma+1}\Bigg].
\end{align*}
\end{lemma}    
\begin{proof}
We want to compute:
\[
\int (Ax + B)^\sigma (C x + D x^2) \, dx
\]

Let:
\[
u = Ax + B \quad \Rightarrow \quad du = A \, dx, \quad dx = \frac{du}{A}, \quad x = \frac{u - B}{A}
\]

Then:
\[
x = \frac{u - B}{A}, \quad x^2 = \left( \frac{u - B}{A} \right)^2
\Rightarrow
C x + D x^2 = C \cdot \frac{u - B}{A} + D \cdot \frac{(u - B)^2}{A^2}
\]

So:
\[
\int (Ax + B)^\sigma (C x + D x^2) \, dx = \frac{1}{A^2} \int u^\sigma \left[ C(u - B) + \frac{D}{A}(u - B)^2 \right] \, du
\]

Now we expand:
\[
C(u - B) + \frac{D}{A}(u^2 - 2B u + B^2)
= C u - C B + \frac{D}{A} u^2 - \frac{2 D B}{A} u + \frac{D B^2}{A}
\]
Multiplying with \( u^\sigma \) yields
\[
u^\sigma\Big[C(u - B) + \frac{D}{A}(u^2 - 2B u + B^2)\Big]= C u^{\sigma+1} - C B u^\sigma + \frac{D}{A} u^{\sigma+2} - \frac{2D B}{A} u^{\sigma+1} + \frac{D B^2}{A} u^\sigma
\]
We can integrate term by term to get
\begin{align*}
&\int (Ax + B)^\sigma (C x + D x^2) \, dx
\\&\quad=\frac{1}{A^2} \left[
\frac{C}{\sigma+2} u^{\sigma+2}
- \frac{C B}{\sigma+1} u^{\sigma+1}
+ \frac{D}{A(\sigma+3)} u^{\sigma+3}
- \frac{2D B}{A(\sigma+2)} u^{\sigma+2}
+\frac{D B^2}{A(\sigma+1)} u^{\sigma+1}
\right].
\end{align*}
Substituting back with \( u = Ax + B \), we obtain
\begin{align*}
&\int (Ax + B)^\sigma (C x + D x^2) \, dx 
\\&\quad=\frac{1}{A^2} \Bigg[
\frac{C}{\sigma+2}(Ax+B)^{\sigma+2}
- \frac{C B}{\sigma+1}(Ax+B)^{\sigma+1}
+ \frac{D}{A(\sigma+3)}(Ax+B)^{\sigma+3}
- \frac{2D B}{A(\sigma+2)}(Ax+B)^{\sigma+2}
\\&\qquad\qquad+ \frac{D B^2}{A(\sigma+1)}(Ax+B)^{\sigma+1}\Bigg]
\\&=\frac{1}{A^2} \Bigg[
\frac{D}{A(\sigma+3)} (Ax + B)^{\sigma+3}
+ \frac{AC-2DB}{A(\sigma+2)}(Ax + B)^{\sigma+2}
+ \frac{B(DB-AC)}{A(\sigma+1)}(Ax + B)^{\sigma+1}\Bigg],
\end{align*}
- not valid for \( \sigma \ne -3, -2, -1 \).
\end{proof}

\begin{theorem}[Reduction to a First-Order ODE using the Integrating Factor method]
\label{thm: main thm}
Let $\eta(\zeta)$ be a twice-differentiable function satisfying the nonlinear ODE
\begin{equation}
\label{eq: ceta6F0}
(A \eta + B)\eta'' + C \eta + D\eta^2 + E (\eta')^2 = 0.    
\end{equation}
Then the derivative $\eta'$ can be expressed in terms of $\eta$ as follows 
\begin{equation}
\label{eq: ceta7}
(\eta')^2=\frac{-2}{A^3}\Bigg[
\frac{D}{\frac{2E}{A}+2} (A\eta + B)^2
+ \frac{AC-2DB}{\frac{2E}{A}+1}(A\eta + B)+ \frac{B(DB-AC)}{\frac{2E}{A}} - \frac{C_1}{(A\eta+B)^{2E/A}}\Bigg],
\end{equation}
where $C_1$ is a constant so that $\eta=0$ implies $\eta'=0$. More precisely,
\begin{equation}
C_1=\Big(\frac{DB^2}{\frac{2E}{A}+2}
+ \frac{(AC-2DB)B}{\frac{2E}{A}+1}+ \frac{B(DB-AC)}{\frac{2E}{A}}\Big)B^{\frac{2E}{A}}.   
\end{equation}
As a consequence, \eqref{eq: ceta6F0} is a conservative system, with the conserved Hamiltonian given by
\[
H(\eta,\eta')=\frac{1}{2}(\eta')^2+\frac{1}{A^3}\Bigg[
\frac{D}{\frac{2E}{A}+2} (A\eta + B)^2
+ \frac{AC-2DB}{\frac{2E}{A}+1}(A\eta + B)\Bigg]+\frac{C_1}{2(A\eta+B)^{2E/A}}.
\]
\end{theorem}

\begin{proof}
Define $v(\eta) := \eta'$. Then, using the chain rule,
\[
\eta'' = \frac{dv}{d\zeta} = \frac{dv}{d\eta} \frac{d\eta}{d\zeta} = v \frac{dv}{d\eta}.
\]
Substituting into the original equation we get
\[
(A\eta + B) v \frac{dv}{d\eta} + C\eta + D\eta^2 + E v^2 = 0.
\]
Multiplying this through by $2$ yields:
\[
2(A\eta + B)v \frac{dv}{d\eta} + 2E v^2 + 2C\eta + 2D\eta^2 = 0.
\]
Since $\frac{d}{d\eta}(v^2) = 2v \frac{dv}{d\eta}$, we get:
\[
(A\eta + B) \frac{d}{d\eta}(v^2) + 2E v^2 + 2C\eta + 2D\eta^2 = 0.
\]
Let $w(\eta) = v^2 = (\eta')^2$, then the above equation can be expressed via $w(\eta)$ as follows
\begin{equation}
\label{eq1}
(A\eta + B)\frac{dw}{d\eta} + 2E w + 2C\eta + 2D\eta^2 = 0.
\end{equation}
This is a first-order linear ODE for $w(\eta)$ of the form:
\begin{equation}
\label{eq: w}
\frac{dw}{d\eta} + p(\eta) w = q(\eta), 
\end{equation}
with
\begin{equation}
\label{eq: p and q}
p(\eta) = \frac{2E}{A\eta + B}, \quad q(\eta) = -\frac{2C\eta + 2D\eta^2}{A\eta + B}.    
\end{equation}
From their formulae we deduce that $A\neq 0$, $E\neq 0$. The integrating factor is given by
\[
\mu(\eta) = \exp\left( \int \frac{2E}{A\eta + B} d\eta \right) = (A\eta + B)^{2E/A}.
\]
Multiplying both sides of \eqref{eq1} by $\mu(\eta)$ gives
\[
\frac{d}{d\eta}[\mu(\eta) w] = \mu(\eta) q(\eta).
\]
Letting
\[
\sigma:=\frac{2E}{A}-1.
\]
Integrating both sides and applying Lemma \ref{lem: integration} we obtain
\begin{align*}
\mu(\eta) w&= \int \mu(\eta) q(\eta)\, d\eta + C_1\\
&=-2\int (A\eta + B)^{2E/A-1}(C\eta+D\eta^2)\, d\eta + C_1
\\&=\frac{-2}{A^2} \Bigg[
\frac{D}{A(\sigma+3)} (A\eta + B)^{\sigma+3}
+ \frac{AC-2DB}{A(\sigma+2)}(A\eta + B)^{\sigma+2}+ \frac{B(DB-AC)}{A(\sigma+1)}(A\eta + B)^{\sigma+1}\Bigg]+C_1.
\end{align*}
Thus we get
\begin{align*}
w(\eta) &= (\eta')^2
\\&= \frac{-2}{A^2}\frac{1}{(A\eta+B)^{\sigma+1}} \Bigg[
\frac{D}{A(\sigma+3)} (A\eta + B)^{\sigma+3}
+ \frac{AC-2DB}{A(\sigma+2)}(A\eta + B)^{\sigma+2}+ \frac{B(DB-AC)}{A(\sigma+1)}(A\eta + B)^{\sigma+1}\Bigg]\notag
\\&\qquad + \frac{C_1}{(A\eta+B)^{\sigma+1}}
\\&=\frac{-2}{A^2}\Bigg[
\frac{D}{A(\sigma+3)} (A\eta + B)^2
+ \frac{AC-2DB}{A(\sigma+2)}(A\eta + B)+ \frac{B(DB-AC)}{A(\sigma+1)}\Bigg]+ \frac{C_1}{(A\eta+B)^{\sigma+1}}
\\&=\frac{-2}{A^2}\Bigg[
\frac{D}{2(E+A)} (A\eta + B)^2
+ \frac{AC-2DB}{2E+A}(A\eta + B)+ \frac{B(DB-AC)}{2E}\Bigg]+ \frac{C_1}{(A\eta+B)^{2E/A}}.
\end{align*}
This is the proposed equation \eqref{eq: ceta7}, thus completing the case $A\neq 0$.
\end{proof}

As a heuristic observation, we note that no dissipation is built into the system; it is therefore conservative with kinetic energy temporarily transferred into the potential energy of the linear springs and subsequently recovered.
\begin{theorem}[Equilibrium Points and Stability]
The equilibrium points of the ODE \cite{Jordan2011nonlinear}occur at values of $\eta = \eta_e$ such that $\eta' = 0$ and $\eta'' = 0$. These satisfy the equation:
\[
C\eta + D\eta^2 = 0,
\]
provided $A\eta + B \neq 0$.
\end{theorem}

\begin{proof}
At equilibrium, $\eta' = 0$ and $\eta'' = 0$. Substituting these into the ODE:
\[
(A\eta + B) \cdot 0 + C\eta + D\eta^2 + E \cdot 0 = 0 \Rightarrow C\eta + D\eta^2 = 0.
\]
This simplifies to:
\[
\eta(C + D\eta) = 0,
\]
with solutions:
\[
\eta = 0, \quad \eta = -\frac{C}{D} \quad \text{(if } D \ne 0).
\]
If $A\eta + B = 0$ at a critical point, the equation may be singular due to the vanishing coefficient of $\eta''$. Such cases require special treatment (e.g., matched asymptotic expansions). In fact this is the case at a bifurcation point of the equations but outside the physical regime.
\end{proof}

It may be shown that

\begin{equation}
\frac{2E}{A}=\frac{c^2-1/2}{c^2-1}
\end{equation}
We can consider the LHS of \eqref{eq: ceta7} as a function of $z:=A\eta+B$, where $z\in (-\infty, B)$ (noting that $B>0$). In term of $z$:
\[
H(z):=\frac{-2}{A^2}\Bigg[
\frac{D}{2(E+A)} z^2
+ \frac{AC-2DB}{2E+A}z+ \frac{B(DB-AC)}{2E}\Bigg]+ \frac{C_1}{z^{2E/A}},
\]
We define
\[
\bar{H}=\bar{H}(c, K,\theta_0)=\min_{z\in (-\infty,B)} H(z).
\]
From Theorem \ref{thm: main thm} we obtain the following phase portrait of the reduced system.

\begin{figure}[htbp]
    \centering
    \includegraphics[width=0.7\textwidth,trim=10 10 10 10,clip]{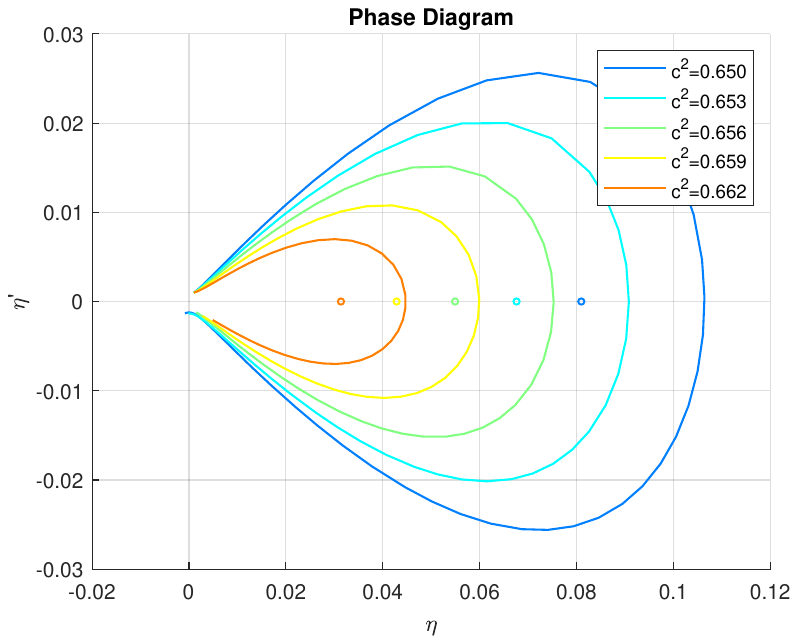}
    \caption{Full numerics}{Homoclinic orbits with centres; obtained using ode45 on coupled first order differential equations.}
    \label{fig 2}
\end{figure}

\begin{theorem}[Phase Portrait of the Reduced System]
The phase portrait of the system is given by:
\[
\left( \frac{d\eta}{d\zeta} \right)^2 = H(A\eta+B),
\]
The trajectories in the phase plane $(\eta, \eta')$ are:
\[
\eta' = \pm \sqrt{H(A\eta+B)}.
\]
Behavior is governed by the sign of $w(\eta)$:
\begin{itemize}
    \item If $H(A\eta+B) > 0$: trajectories exist with nonzero slope.
    \item If $H(A\eta+B) = 0$: critical points or turning points.
    \item If $H(A\eta+B) < 0$: solutions are nonreal — dynamically forbidden regions. In particular, the forbidden region is given by 
    \[
    \{(c,K,\theta_0): \bar{H}(c,K,\theta_0)<0\}.
    \]
\end{itemize}
\end{theorem}
\section{Numerical investigations}
\label{sec: numerics}
A necessary condition for the existence of a soliton solution is that  $C<0$, that is
\[
P=\frac{4(1-K)}{1-c^2}\sin(\theta_0)^2-8(K+\sin(\theta_0)^2)<0.
\]
Solving for this condition gives (recalling that $0<c<1$)
\[
c^2<1-\frac{4(1-K)\sin(\theta_0)^2}{8(K+\sin(\theta_0)^2)}\approx 0.671.
\]
Numerical solutions for this equation, using ODE45 in Matlab, are shown in Figure \ref{fig: 12solutions}. We observe that the range of $c^2$ for the existence of cnoidal (Jacobian elliptic function) solutions is
\[
c^2\geq 0.671.
\]
Numerical results indicate that soliton-bearing solutions exist within the approximate range.

\begin{equation*}
    5/8 \lesssim c^2 \lesssim 2/3,
\end{equation*}
which implies the corresponding parameter regime:
\begin{equation*}
    1/3 \lesssim -2E/A \lesssim 1/2.
\end{equation*}

\begin{figure}[htbp]
    \centering
    \begin{minipage}[b]{0.6\textwidth}
        \includegraphics[width=\linewidth]{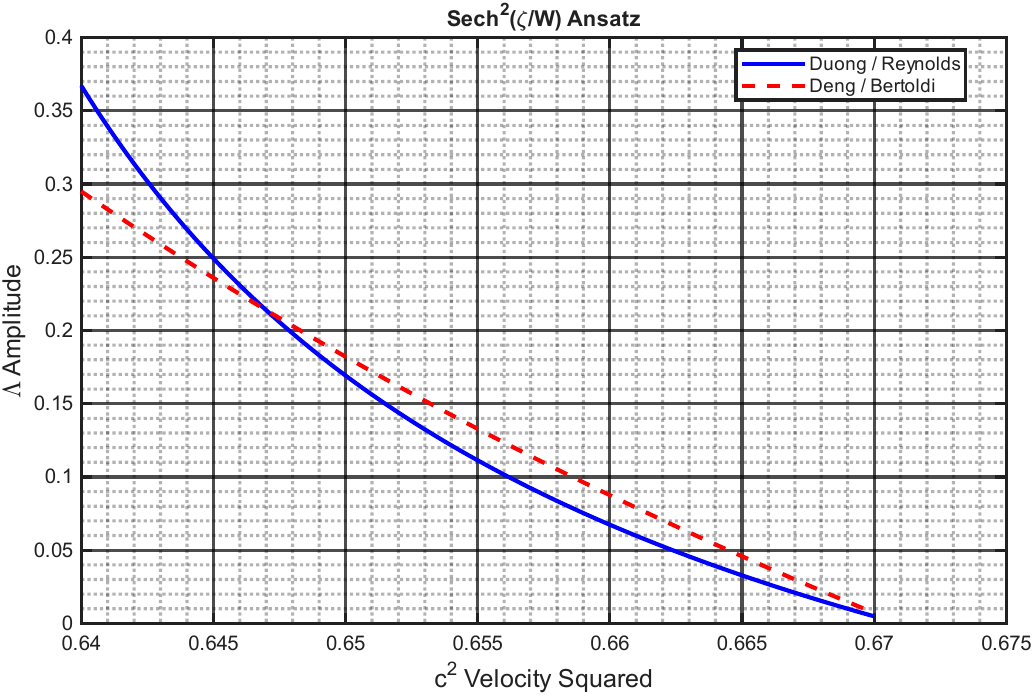}
    \end{minipage}
    \begin{minipage}[b]{0.6\textwidth}
        \includegraphics[width=\linewidth]{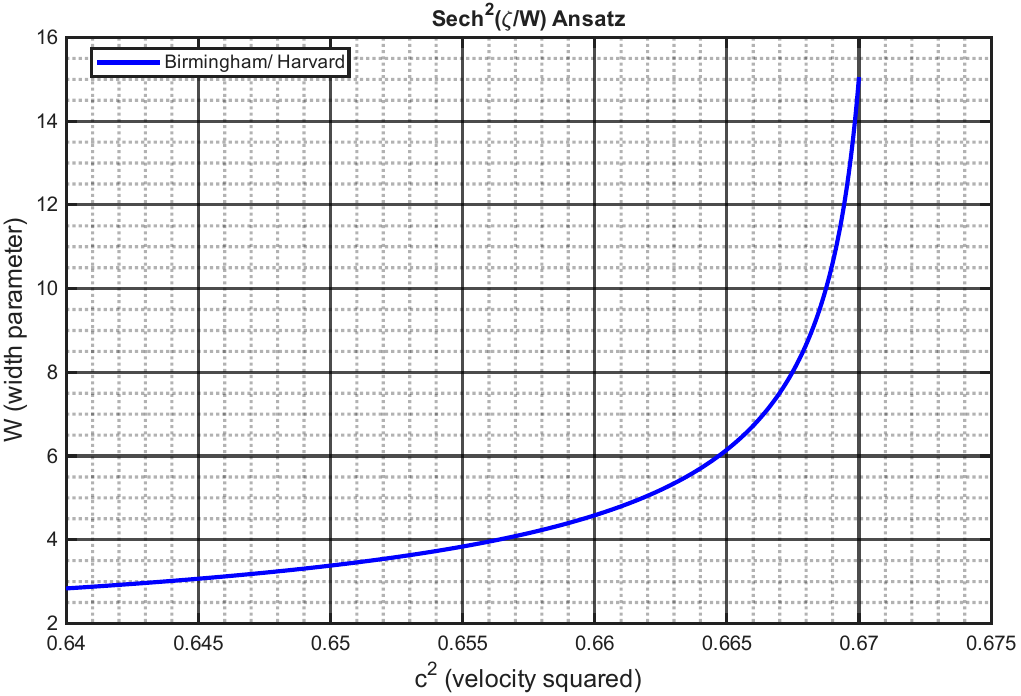}
    \end{minipage}
\end{figure}
    
\begin{figure}[htbp]
    \centering
    \begin{minipage}[b]{0.24\textwidth}
        \includegraphics[width=\linewidth]{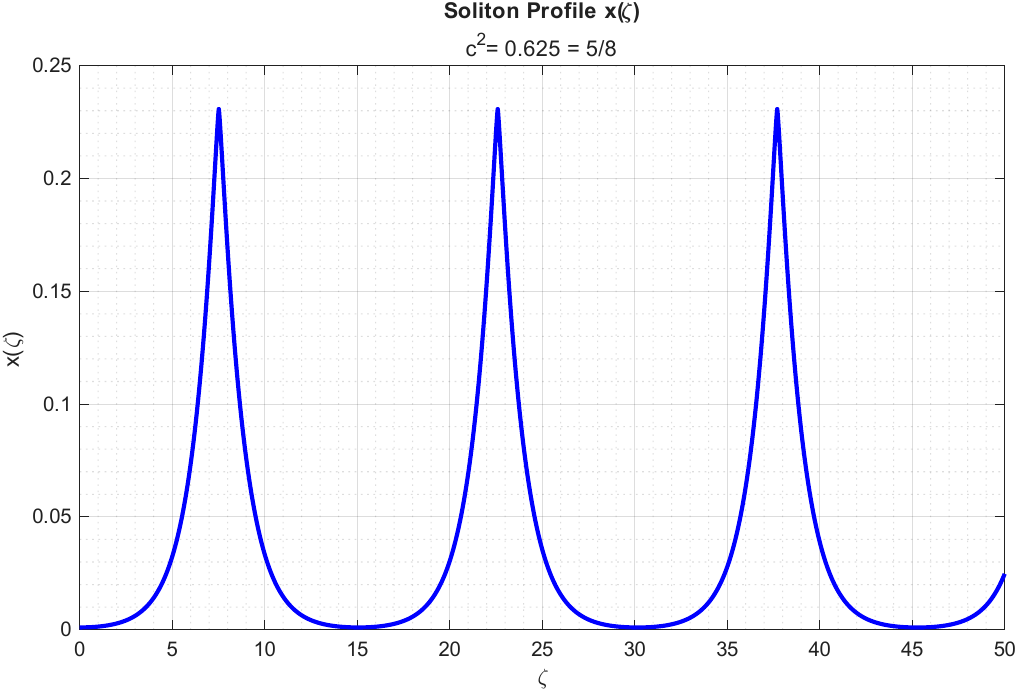}
        \subcaption{$c^2=0.625$}
    \end{minipage}
    \begin{minipage}[b]{0.24\textwidth}
        \includegraphics[width=\linewidth]{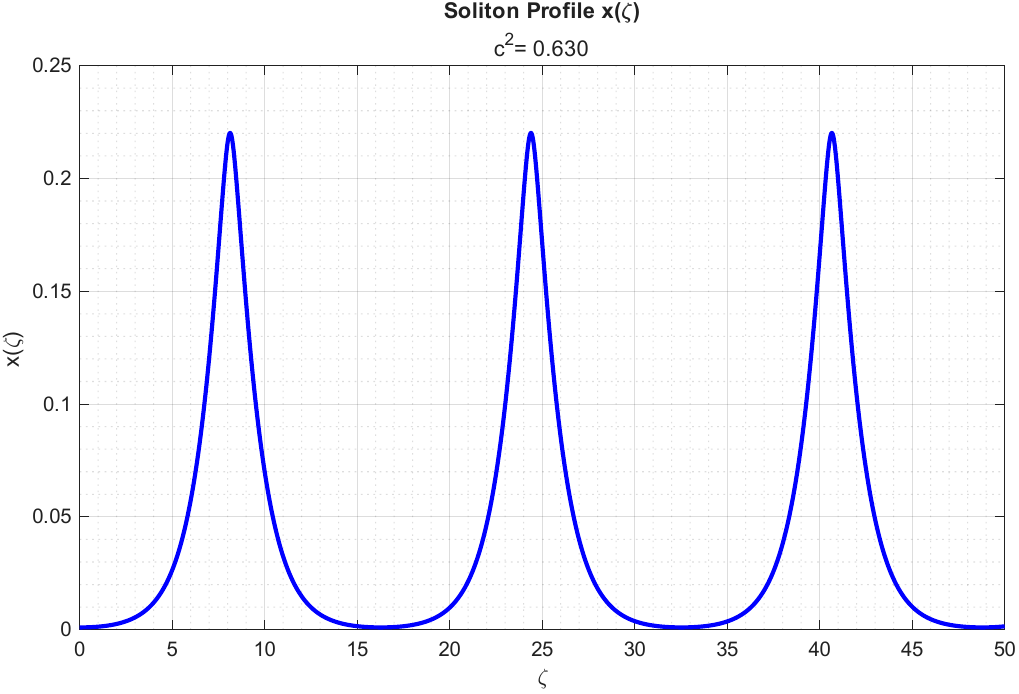}
        \subcaption{$c^2=0.630$}
    \end{minipage}
    \begin{minipage}[b]{0.24\textwidth}
        \includegraphics[width=\linewidth]{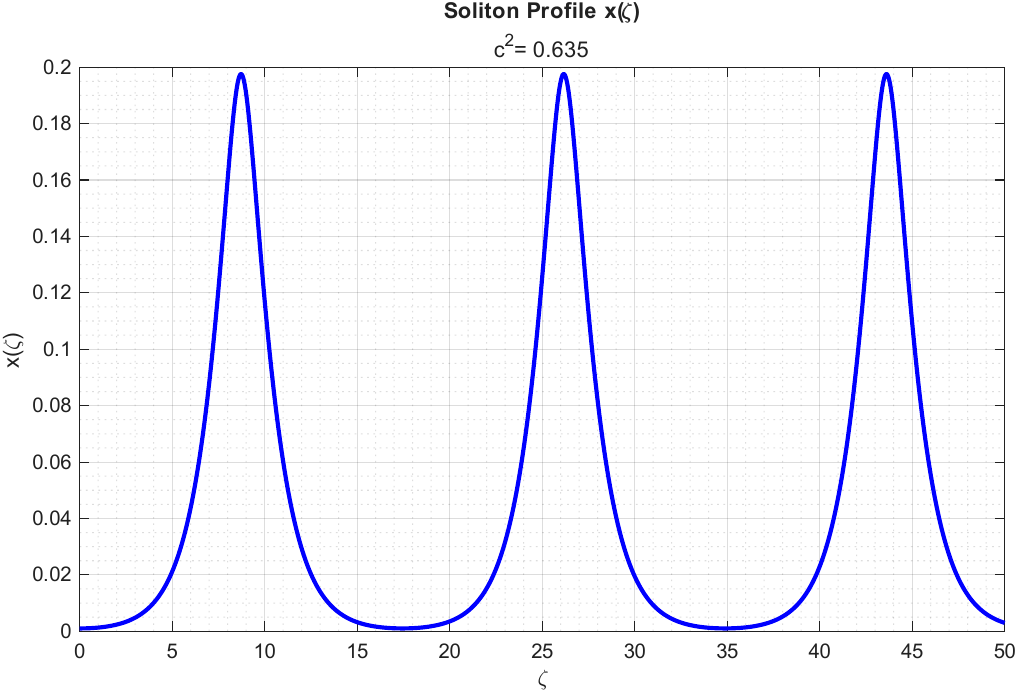}
        \subcaption{$c^2=0.635$}
    \end{minipage}

    \vspace{0.5em}
    
    \begin{minipage}[b]{0.24\textwidth}
        \includegraphics[width=\linewidth]{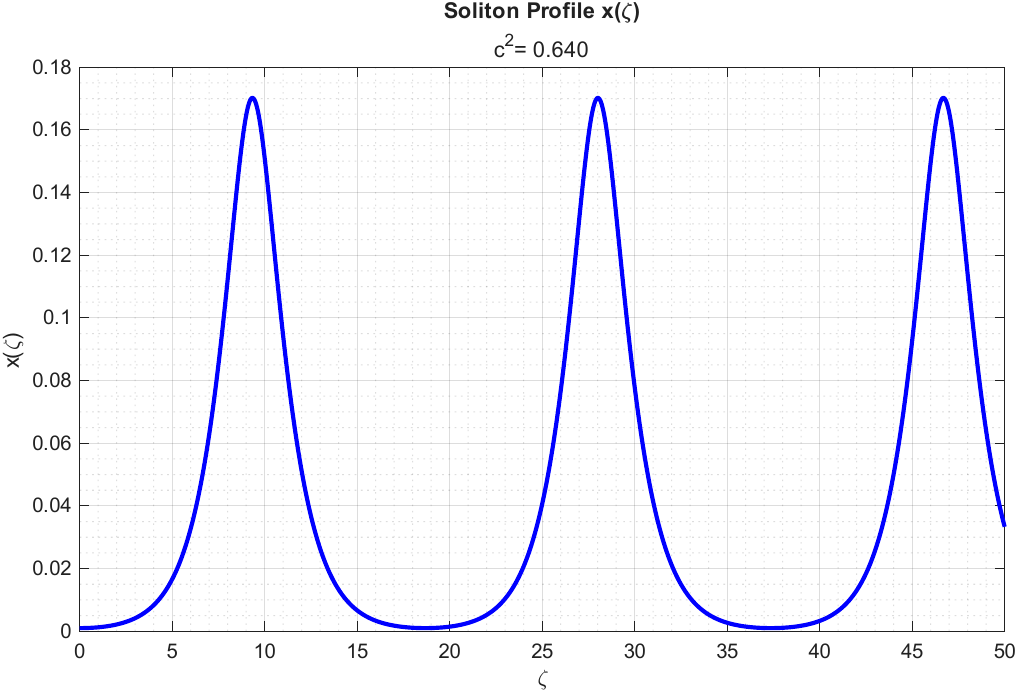}
        \subcaption{$c^2=0.640$}
    \end{minipage}
    \begin{minipage}[b]{0.24\textwidth}
        \includegraphics[width=\linewidth]{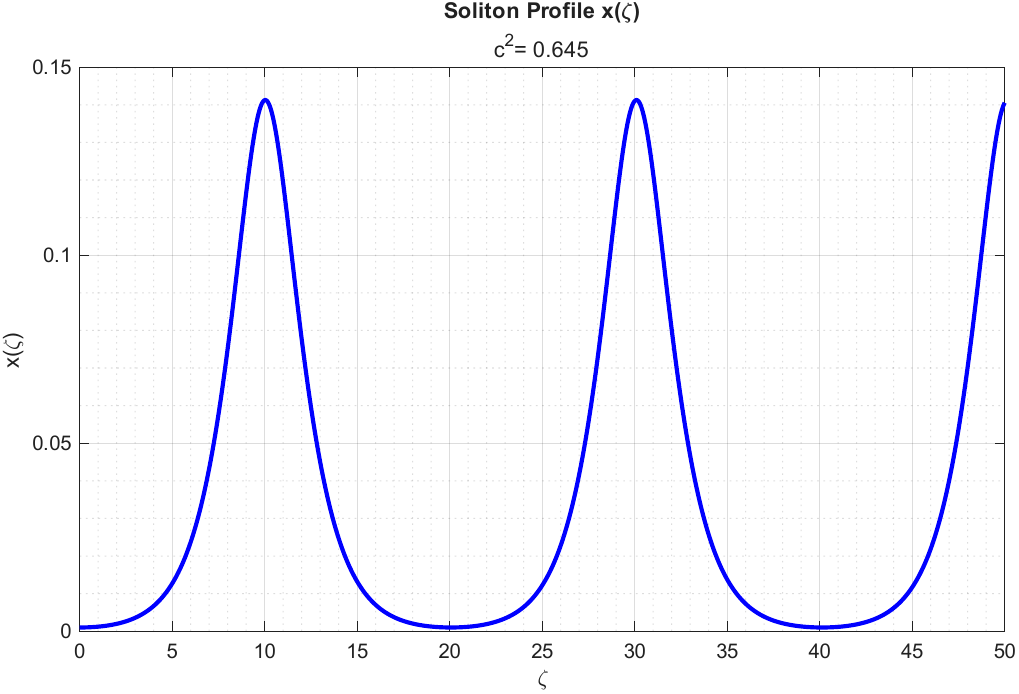}
        \subcaption{$c^2=0.645$}
    \end{minipage}
    \begin{minipage}[b]{0.24\textwidth}
        \includegraphics[width=\linewidth]{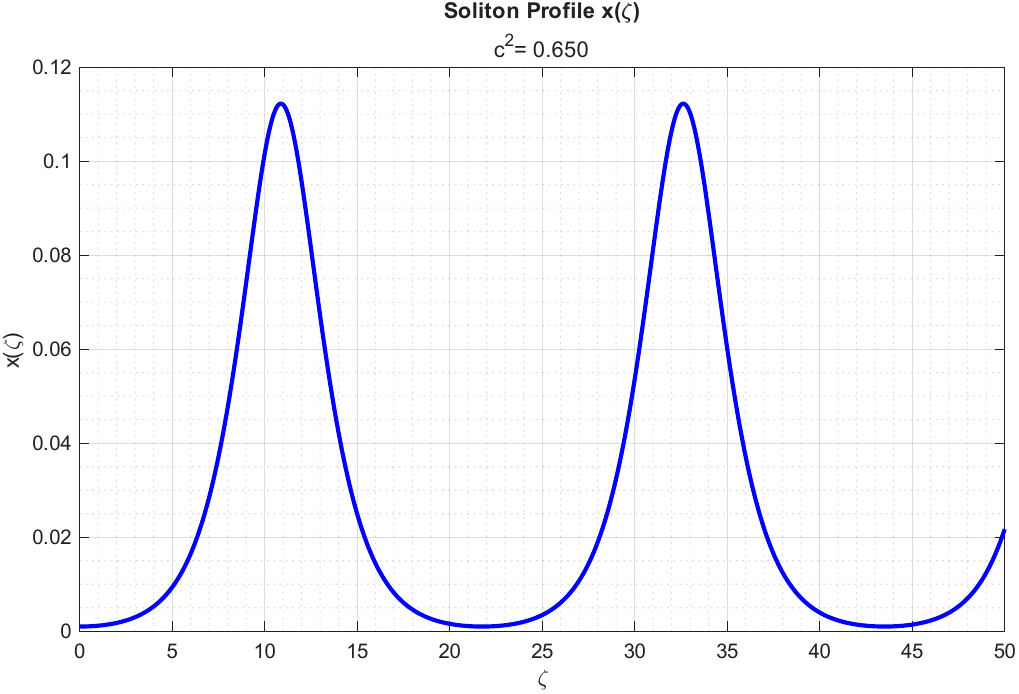}
        \subcaption{$c^2=0.650$}
    \end{minipage}

    \vspace{0.5em}
    
    \begin{minipage}[b]{0.24\textwidth}
        \includegraphics[width=\linewidth]{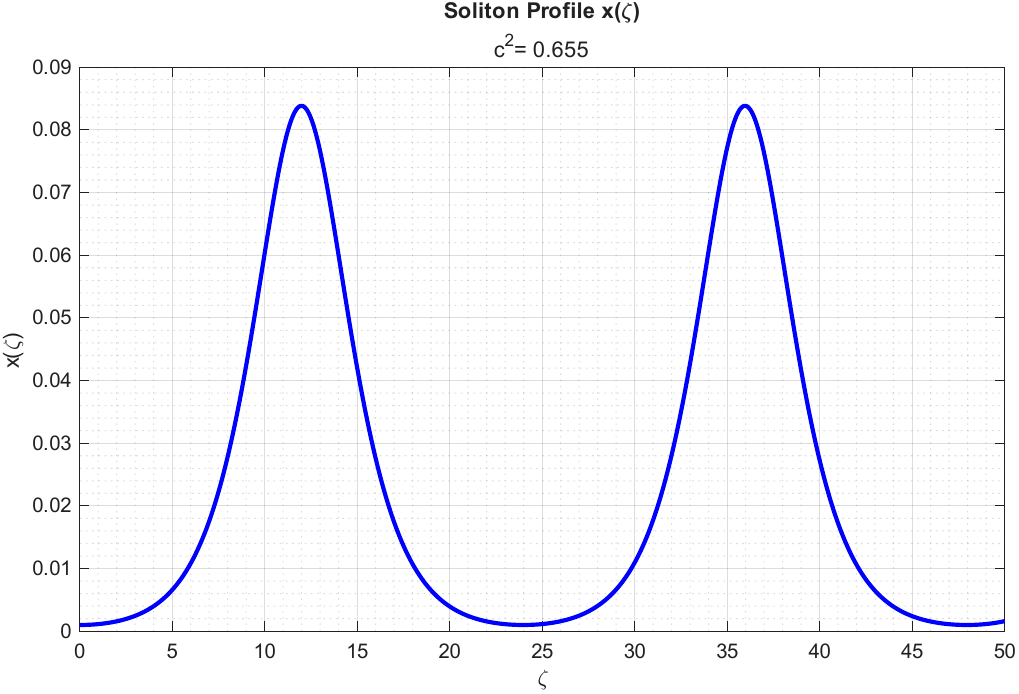}
        \subcaption{$c^2=0.655$}
    \end{minipage}
    \begin{minipage}[b]{0.24\textwidth}
        \includegraphics[width=\linewidth]{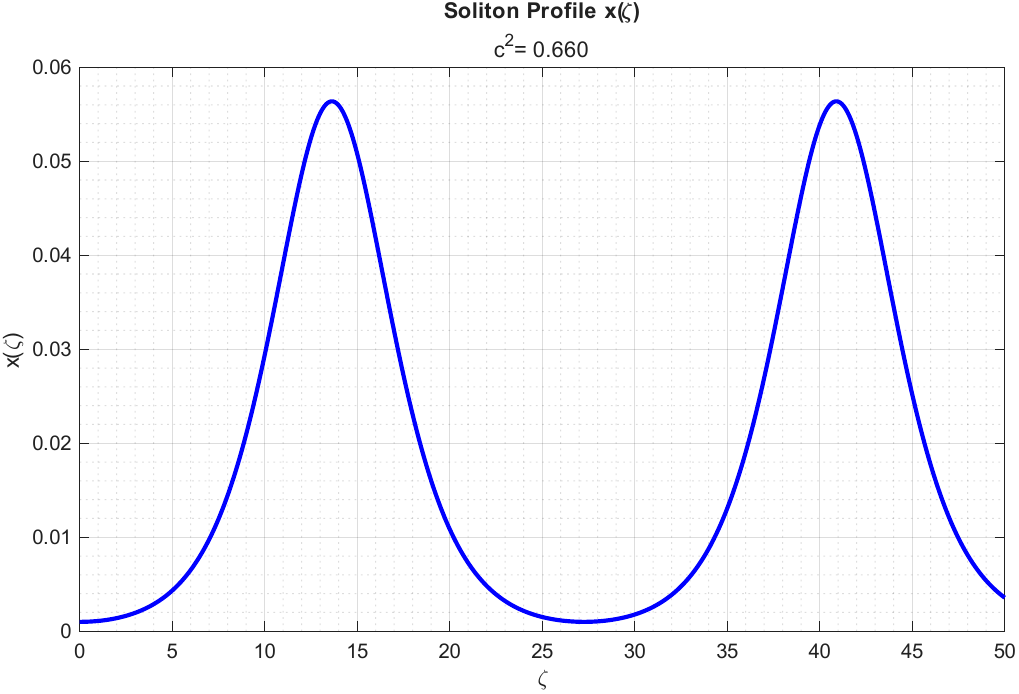}
        \subcaption{$c^2=0.660$}
    \end{minipage}
    \begin{minipage}[b]{0.24\textwidth}
        \includegraphics[width=\linewidth]{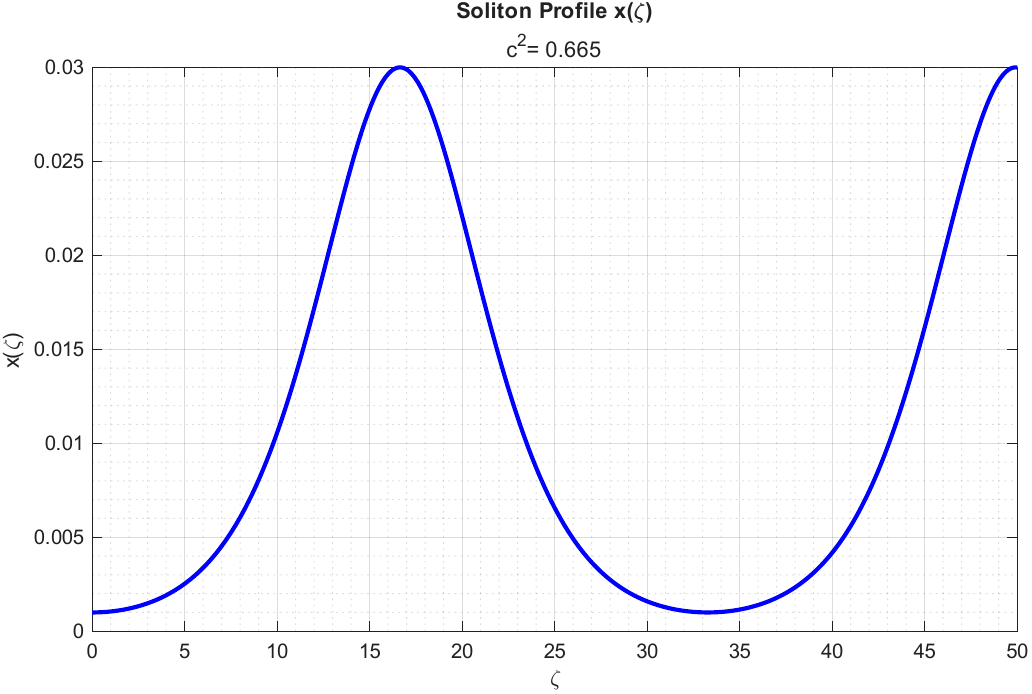}
        \subcaption{$c^2=0.665$}
    \end{minipage}

    \vspace{0.5em}

    \begin{minipage}[b]{0.24\textwidth}
        \includegraphics[width=\linewidth]{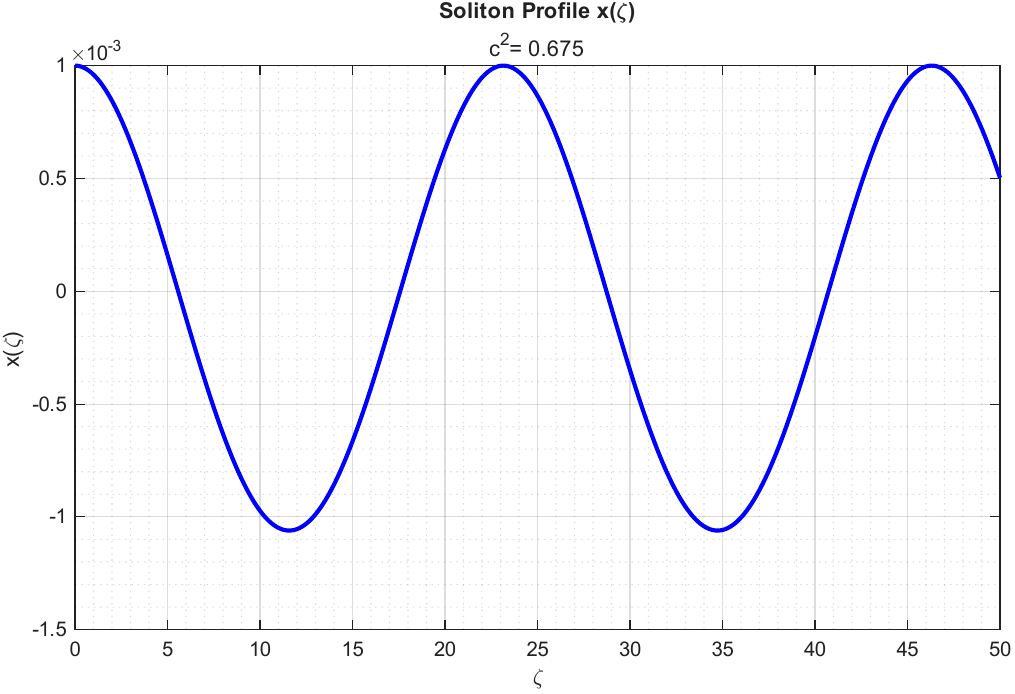}
        \subcaption{$c^2=0.670$}
    \end{minipage}
    \begin{minipage}[b]{0.24\textwidth}
        \includegraphics[width=\linewidth]{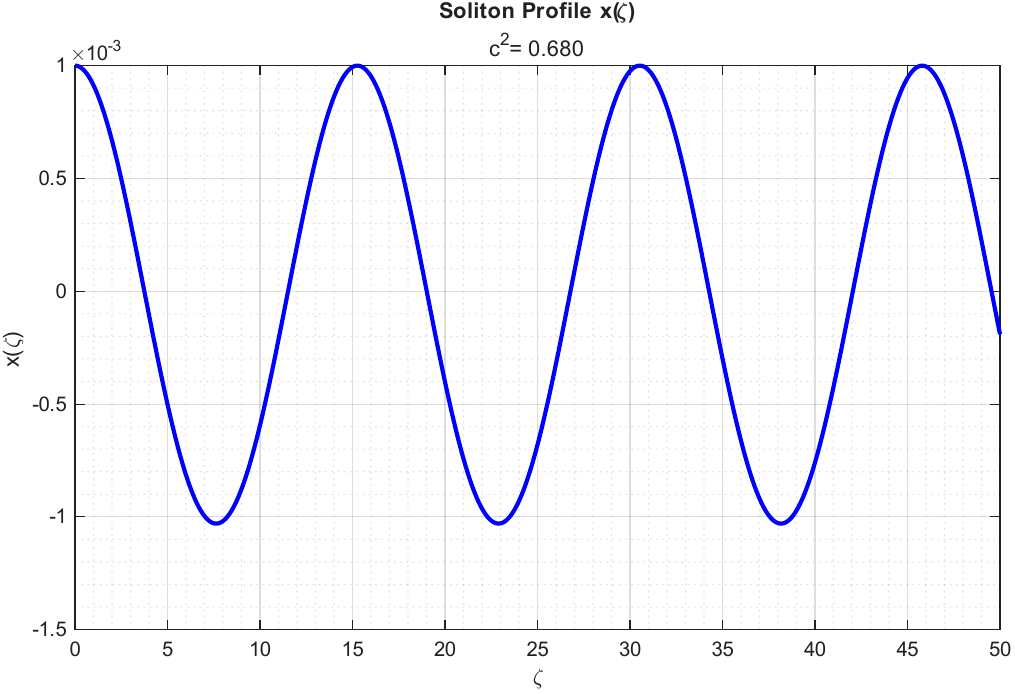}
        \subcaption{$c^2=0.675$}
    \end{minipage}
    \begin{minipage}[b]{0.24\textwidth}
        \includegraphics[width=\linewidth]{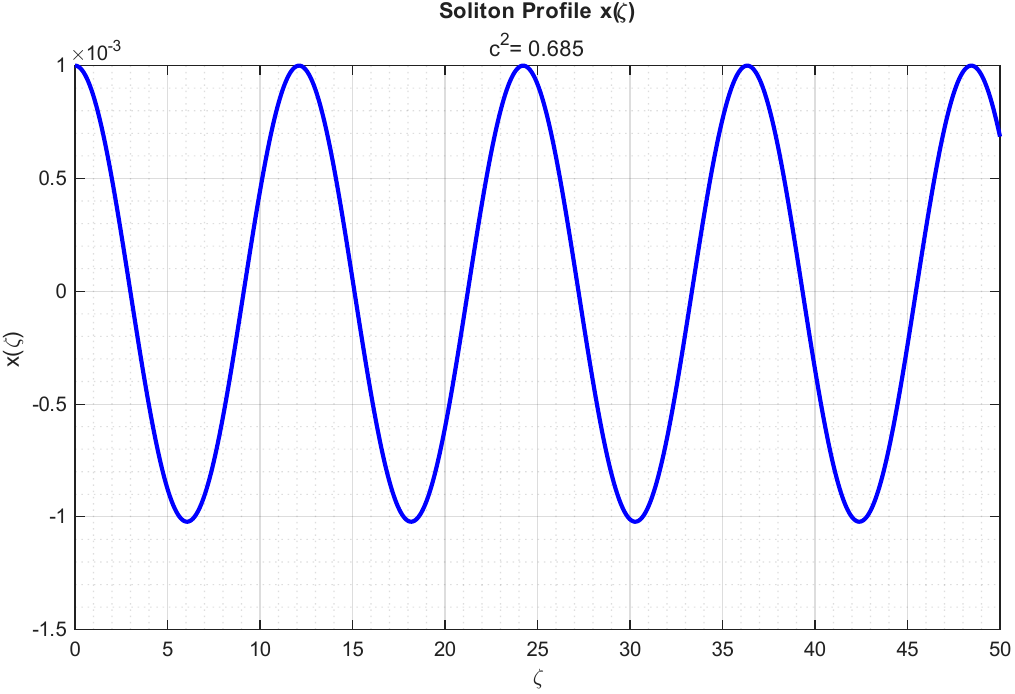}
        \subcaption{$c^2=0.680$}
    \end{minipage}

    \caption{Numerical solutions for various values of $c^2$}
\label{fig: 12solutions}
\end{figure}

The bounds are not exact, but they provide a useful rational approximation that closely brackets the soliton regime. At the lower end of this interval, the solution exhibits a sharply peaked, spike-like profile (a peakon). Zooming into the vicinity of the spike (see Figure \ref{fig: 12solutions}), we find that whilst the second derivative becomes very large, it remains finite. This indicates that the solution is continuous and differentiable, but with very high curvature at the peak - which suggests the presence of a narrow but smooth transition zone, rather than a mathematical cusp or singularity.

It is worth noting that the sharp spike observed near the cut-off $c^2 \approx5/8^{-}$ is a feature of the continuum approximation of the system. In the discrete model from which it is derived, such a narrow, high-curvature structure only has a Full Width Half Maximum of two sites, and therefore cannot be interpreted as a physically meaningful solution.

This indicates that the spike-like solution is more likely an artifact of the approximation, rather than a true feature of the underlying discrete dynamics. The continuum model, whilst analytically valuable, begins to lose fidelity near this limit - especially as the solution width approaches a lattice spacing.\\

\section{A direct comparison with \cite{deng2017elastic}}
\label{sec: comparison}
The soliton bearing equation obtained in \cite{deng2017elastic} can be written:
\begin{equation}
\label{eq: Deng}
    B\eta'' + C\eta + D^{*}\eta^2=0, 
\end{equation}
where
\begin{equation}
    D^{*}=2D/3.
\end{equation}
Recall that ours is
\begin{equation}
(A \eta + B)\eta'' + C \eta + D\eta^2 + E (\eta')^2 = 0,
\end{equation} 
where all the coefficients are as defined in \eqref{eq: coefficients}.
By noticing that \eqref{eq: Deng} is none other than the Klein-Gordon equation, so \cite{deng2017elastic} introduced the ans\"atz:
\begin{equation}
    \eta=\Lambda \mathrm{sech}^2(\zeta/W),
\end{equation}
where the amplitude $\Lambda$ and width $W$ are both functions of $c^2$. Using again this ans\"atz, direct computation gives
\begin{equation}
    \eta'^2=\frac{4}{W^2}\left(\eta^2-\frac{\eta^3}{\Lambda}\right),
\end{equation}
and hence
\begin{equation}
    \eta''=\frac{1}{2}\frac{d}{d\eta}\eta'^2=\frac{2}{W^2}\left(2\eta-\frac{3\eta^2}{\Lambda}\right).
\end{equation}
Comparing like with like between soliton solutions of \cite{deng2017elastic} and ours, we only need balance the coefficients to second order. The first-order term multiplying $\eta$ is the same as in \cite{deng2017elastic}:
\begin{equation}
\label{eq: first}
    \frac{4B}{W^2}+C=0,
\end{equation}
but for second-order term multiplying $\eta^2$, \cite{deng2017elastic} has
\begin{equation}
\label{eq: 2ndorderDeng}
    \frac{-6B}{W^2\Lambda}+D^{*}=0,
\end{equation}
and we have
\begin{equation}
\label{eq: 2ndorder ours}
    \frac{4A}{W^2}-\frac{6B}{W^2\Lambda}+D+\frac{4E}{W^2}=0.
\end{equation}
The first order equation \eqref{eq: first} gives width as a function of $c^2$. To proceed we use the width from the lowest order approximation and substitute it into each subsequent expression for $\Lambda_{n+1}$. So we get a continually refined solution:

\begin{equation*}
    W^2_0=f_0(c^2),\\
\end{equation*}

\begin{equation*}
    \Lambda_{n+1}=g_{n+1}(W^2_0, c^2),
\end{equation*}
for $n \in \mathbb{N}_0$ and in \cite{deng2017elastic}, the approximation scheme gives

\begin{equation*}
    \Lambda_1=\frac{-3C}{2D^{*}},
\end{equation*}
and in ours
\begin{equation*}
    \Lambda_1=\frac{3}{2}\left(\frac{1}{B}(A+E)-\frac{D}{C}\right)^{-1}.
\end{equation*}
We notice that our formulas will reduce to \cite{deng2017elastic} if one ignores $A$ and $E$. The results are compared in Figure 1.

Reading off from the graph at $\Lambda=0.05$, \cite{deng2017elastic} found the numerical solution of the discrete equations gave $c^2=0.6448$, and the continuum model plus $\mathrm{sech}^2$ ans\"atz gave $c^2=0.6646$.\\

On the other hand our continuum model gave $c^2=0.6624$. So whilst the model in \cite{deng2017elastic} compared well to the numerics, with an overshoot of only $3.069\%$, ours was an overshoot by $2.728\%$ - closer by exactly $1/9$. 

The numerical solution of our continuum approximation without assumptions about the shape of the waveform (see Figure 1), leads to $c^2\approx0.6615$ for an amplitude of 0.05. This is very close indeed to the ans\"atz result.

It is evident from our numerical work that cnoidal solutions are going to become important, but these require equations which are fully third-order and will be the topics of future works.

\begin{figure}
    \centering
    \includegraphics[width=0.5\linewidth]{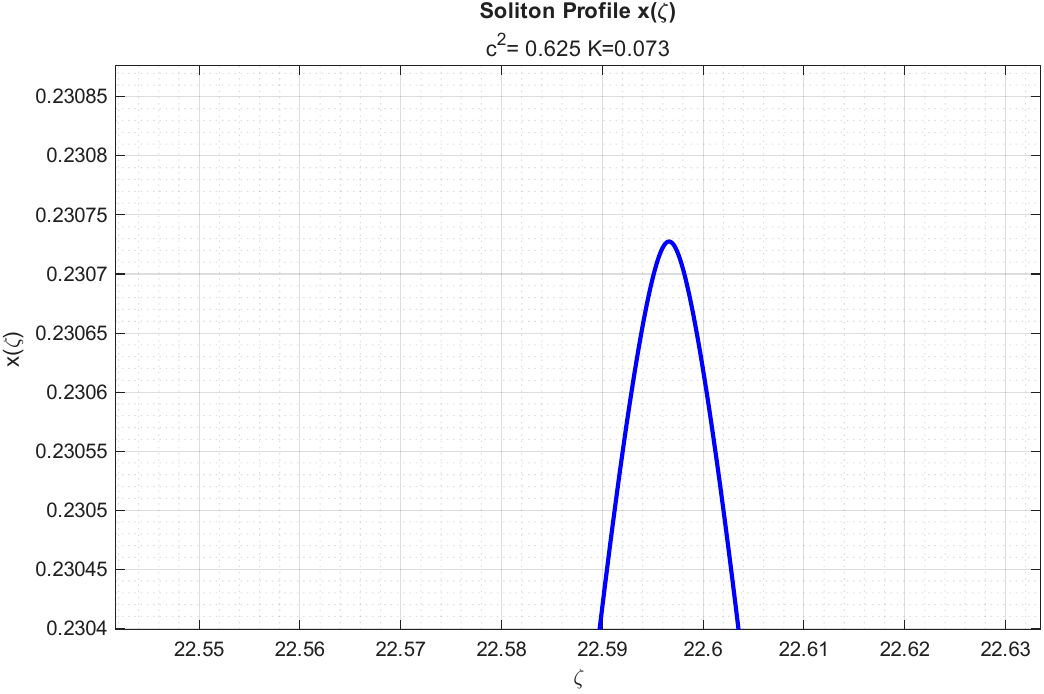}
    \caption{Close-up of a spike at $c^2=0.625$}
    \label{fig:8}
\end{figure}

\section*{Acknowledgment}
M. H. D is funded by an EPSRC Standard Grant EP/Y008561/1.\\

M. J. R has an interest in nonlinear waves going back to the age of around 12. His family used to holiday each year at Silverdale in Lancashire. Not far away was the estuary of the River Kent which at certain times of the year witnessed a small tidal bore (a solitary wave). At that age he didn't understand it much but he gave it a lot of thought. He wishes to dedicate the paper to his late father Revd. Raymond Ernest Reynolds S.Th. M.Phil. (Nottingham) who passed away last year. He passionately believed in a purposeful life, where we \textit{build the is with the bricks of the might have been.} Aged 26 M. J. R began doctoral studies with Prof. E. G. Wilson of the Physics Department at QMUL, on Wilson's 1983 theoretical physics paper on polarons - effectively solitary waves in an inhomogeneous medium. \cite{Wilson83}. Now aged nearly 64 he wishes to bring it full circle back to the child who simply wondered.

\section{Appendix: detailed derivation of the discrete model}
\label{sec: appendix}
Let us summarise the numerics of \cite{deng2017elastic} square-lattice model:
\begin{enumerate}
    \item m (mass) = $2.093g$ 
    \item J (moment of inertia) = $18.11g$ $mm^{2}$
    \item $2^{1/2}l$ (side length) = $8mm$
    \item $\theta_0$ (pre-deformation) = 25$^0$
    \item r (radius of copper cylinder) = $2.38mm$
    \item $k_l$ (longitudinal spring stiffness) = $19,235Nm^{-1}$
    \item $k_\theta$ (torsional spring stiffness) = $4.27$ x $10^{-2}Nm$ $rad^{-1}$
\end{enumerate}

The squares are considered rigid, but are connected by thin and highly deformable ligaments where all the deformation takes place. The system can support elastic vector solitons due to the mode coupling governed by the pre-deformation angle $\theta_0$.

Deng modelled the system by assuming no damping, vertically periodic boundary conditions, free boundary conditions at the right hand side, all waves travelling in the x-direction from left to right and having no DOF in the y-direction (ie a column of squares will have the same horizontal displacement and will rotate by the same angle up to a sign convention), and finally neighbouring squares always rotate in an opposite sense.

I would like to introduce new notation by considering the bank of squares as a matrix. We can say that $u^{i,j}$ = $u^{i+1,j}$, and $\theta^{i,j}$ = $\theta^{i+1,j}$ under a suitable sign convention.

Focussing on the $[i,j]^{th}$ square we can write the equations of motion as follows:

\begin{align*}
    m \ddot{u}^{i,j}=\sum_{p=1}^{4}F_p^{i,j},
\end{align*}

\begin{align*}
    J \ddot{\theta}^{i,j}=\sum_{p=1}^{4}M_p^{i,j}.
\end{align*}

where $F_p^{i,j}$ is the force in the x direction, and $M_p^{i,j}$ is the moment generated about the centre at the $p^{th}$ vertex.

It is convenient to work with vectors relative to the centre of each square: $\mathbf{r}_1(\theta^{i,j})$, $\mathbf{r}_2(\theta^{i,j})$, $\mathbf{r}_3(\theta^{i,j})$, and $\mathbf{r}_4(\theta^{i,j})$. 

Now when thinking about the rotational sign convention we only need to notate a unit cell and the pattern repeats. Using $q$ = 1 or 2 for each rigid square we have

\begin{align*}
    \mathbf{r}_1(\theta^{i,j})=lcos(\theta^{i,j}+\theta_0)\hat{\mathbf{e_x}}+l(-1)^qsin(\theta^{i,j}+\theta_0)\hat{\mathbf{e_y}},
\end{align*}

\begin{align*}
    \mathbf{r}_2(\theta^{i,j})=-l(-1)^qsin(\theta^{i,j}+\theta_0)\hat{\mathbf{e_x}}+lcos(\theta^{i,j}+\theta_0)\hat{\mathbf{e_y}},
\end{align*}

\begin{align*}
    \mathbf{r}_3(\theta^{i,j})=-lcos(\theta^{i,j}+\theta_0)\hat{\mathbf{e_x}}-l(-1)^qsin(\theta^{i,j}+\theta_0)\hat{\mathbf{e_y}},
\end{align*}

\begin{align*}
    \mathbf{r}_4(\theta^{i,j})=l(-1)^qsin(\theta^{i,j}+\theta_0)\hat{\mathbf{e_x}}-lcos(\theta^{i,j}+\theta_0)\hat{\mathbf{e_y}}.
\end{align*}

Now focusing on one ligament connecting two neighbouring vertices, we can make a novel diagrammatic representation to isolate each DOF and the coupling between them. There are four steps of causation from an impulse coming in from the left hand side:

\begin{itemize}
    \item a plane wave in the x-direction causes compression to spread throughout the system, and this will be resisted as in any elastic media
    \item because of the pre-deformation, the compression wave has a component which will cause rotation of the squares
    \item the ligament will bend creating torsion which tries to return the system to its equilibrium orientation
    \item this in turn will have a component in the x-direction which will counteract the compression
\end{itemize}

\begin{figure}
    \centering
    \includegraphics[width=0.5\linewidth]{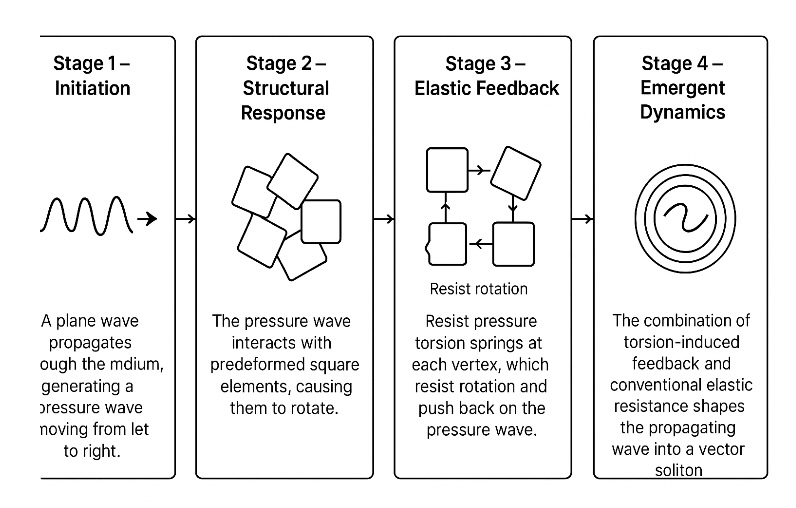}
    \caption{Causal flow chart}
    \label{fig:8}
\end{figure}

We propose to develop this as a Feynman-type diagram so that the situation is clearly visible.

One final thing to mention is the y-dependence of the waves. We have to remember that this is an Auxetic metamaterial in which compression waves are accompanied by a lateral narrowing. At the peak of the soliton the contraction in dimension of the unit cell is:

\begin{equation}
    \Delta=\left(1-\frac{\cos(35^{\circ})}{\cos(25^{\circ})}\right)=0.09617=9.617\%.
\end{equation}
In such a case we need to add linear shear springs as a third element and indeed this has been done several times since. We hope to have added something practical to the discussion around these models.

\newpage

\bibliographystyle{plain}
\bibliography{refs.bib}

\begin{thebibliography}{DRTB17b}

\bibitem[ACO24]{ariza2024homogenization}
MP~Ariza, S~Conti, and M~Ortiz.
\newblock Homogenization and continuum limit of mechanical metamaterials.
\newblock {\em Mechanics of Materials}, 196:105073, 2024.

\bibitem[CR15]{ciattoni2015nonlocal}
Alessandro Ciattoni and Carlo Rizza.
\newblock Nonlocal homogenization theory in metamaterials: Effective electromagnetic spatial dispersion and artificial chirality.
\newblock {\em Physical Review B}, 91(18):184207, 2015.

\bibitem[DRTB17a]{deng2017elastic}
Bolei Deng, JR~Raney, Vincent Tournat, and Katia Bertoldi.
\newblock Elastic vector solitons in soft architected materials.
\newblock {\em Physical review letters}, 118(20):204102, 2017.

\bibitem[DRTB17b]{deng2017supplemental}
Bolei Deng, JR~Raney, Vincent Tournat, and Katia Bertoldi.
\newblock Supplemental material http://link.aps.org/supplemental/10.1103/physrevlett.118.204102.
\newblock {\em Physical Review Letters}, 2017.

\bibitem[GE00]{grima2000auxetic}
JN~Grima and KE~Evans.
\newblock Auxetic behaviour from rotating squares.
\newblock {\em Journal of Materials Science Letters}, 19:1563--1565, 2000.

\bibitem[JS07]{Jordan2007sourcebook}
DW~Jordan and P~Smith.
\newblock {\em Nonlinear Ordinary Differential Equations: Problems and Solutions}.
\newblock OUP, 2007.

\bibitem[JS11]{Jordan2011nonlinear}
DW~Jordan and P~Smith.
\newblock {\em Nonlinear Ordinary Differential Equations}.
\newblock OUP, 2011.

\bibitem[SKG16]{sridhar2016homogenization}
Ashwin Sridhar, Varvara~G Kouznetsova, and Marc~GD Geers.
\newblock Homogenization of locally resonant acoustic metamaterials towards an emergent enriched continuum.
\newblock {\em Computational mechanics}, 57(3):423--435, 2016.

\bibitem[Ste22]{steckiewicz2022homogenization}
Adam Steckiewicz.
\newblock Homogenization of the vertically stacked medium frequency magnetic metamaterials with multi-turn resonators.
\newblock {\em Scientific Reports}, 12(1):20333, 2022.

\bibitem[Wil83]{Wilson83}
EG~Wilson.
\newblock A new theory of acoustic solitary-wave polaron motion.
\newblock {\em Journal of Physics C}, 16:6739--6755, 1983.

\end{thebibliography}


\end{document}